\documentclass[12pt]{article}
\usepackage{mathpazo}
\usepackage{lineno}
\usepackage{amssymb}
\usepackage{graphicx}
\usepackage{amsfonts}
\usepackage{amsmath}
\usepackage{mathrsfs}
\usepackage{amsthm}
\usepackage{mathtools}
\usepackage{ascmac}
\usepackage{bm}
\usepackage{color}
\usepackage{pgfplots}
\pgfplotsset{compat=1.16}
\usepackage{here}
\usepackage{multirow}
\usepackage[hang,small,bf]{caption}
\usepackage[subrefformat=parens]{subcaption}
\usepackage{authblk}
\usepackage[colorlinks=true, citecolor=teal, linkcolor=teal, urlcolor=teal]{hyperref}
\usepackage[authoryear]{natbib}
\usepackage{threeparttable}
\usepackage{comment}
\usepackage{stackengine}
\usepackage{tikz}
\usetikzlibrary{positioning, calc}
\usetikzlibrary{arrows.meta}
\usetikzlibrary{bending}
\usepackage{setspace}

\def\E#1{\mathbb{E}\left[#1\right]}
\def\V#1{\mathbb{V}\left[#1\right]}
\def\P#1{\mathbb{P}\left[#1\right]}

\newcommand{\argmin}{\mathop{\rm arg\,min}\limits}
\newcommand{\esssup}{\mathop{\rm ess\,sup}\limits}

\numberwithin{equation}{section}

\newtheorem{theoremx}{Theorem}
\newenvironment{theorem}
  {\pushQED{\qed}\theoremx}
  {\popQED\endtheoremx}
\newtheorem{assumption}{Assumption}
\newtheorem{assL}{Assumption}

\newtheorem{assF}{Assumption}

\newtheorem{corx}{Corollary}[section]
\newenvironment{corollary}
  {\pushQED{\qed}\corx}
  {\popQED\endcorx}

\newtheorem{lemmax}{Lemma}
\newenvironment{lemma}
  {\pushQED{\qed}\lemmax}
  {\popQED\endlemmax}

\theoremstyle{definition}
\newtheorem{remarkx}{Remark}
 \newenvironment{remark}
  {\pushQED{\qed}\remarkx}
  {\popQED\endremarkx}

\makeatletter
\renewenvironment{proof}[1][\proofname]{%
  \par\pushQED{\qed}\normalfont%
  \topsep6\p@\@plus6\p@\relax
  \trivlist\item[\hskip\labelsep\bfseries#1\@addpunct{.}]%
  \ignorespaces
}{%
  \popQED\endtrivlist\@endpefalse
}
\makeatother



\makeatletter
\renewcommand*{\@fnsymbol}[1]{\ensuremath{\ifcase#1\or \flat\or * \else\@ctrerr\fi}}
\makeatother

\title{\bf Kernel Choice Matters for Local Polynomial Density Estimators at Boundaries}
\author{Shunsuke Imai\thanks{\href{mailto:imai.shunsuke.57n@st.kyoto-u.ac.jp}{imai.shunsuke.57n@st.kyoto-u.ac.jp}}}
\author{Yuta Okamoto\thanks{
\href{mailto:okamoto.yuuta.57w@st.kyoto-u.ac.jp}{okamoto.yuuta.57w@st.kyoto-u.ac.jp} (corresponding author)}}

\affil{Graduate School of Economics, Kyoto University}

\begin{document}
\maketitle
\begin{abstract}
    Local polynomial density (LPD) estimators are widely used for inference on boundary features of the density function. Contrary to conventional wisdom, we show that kernel choice substantially affects efficiency. Theory, simulations, and empirical evidence indicate that the popular triangular kernel delivers large mean squared error, wide confidence intervals, and limited power for detecting discontinuities. Moreover, small-sample variance can explode because the finite-sample variance is infinite under compactly supported kernels. As a simple yet powerful remedy, we recommend using the Gaussian or Laplace kernels. These alternatives yield marked efficiency gains and eliminate variance explosions, improving the reliability of LPD-based inference.
\end{abstract}




\newpage
\doublespacing
\section{Introduction}\label{sec: introduction}
The probability density function at boundary points is often of interest in empirical analyses across economics, political science, and applied statistics.
For instance, certain empirical tests of economic theory rely on density at boundary points (e.g., \citealp{Collin_Talbot:2023}), and discontinuity detection has been used to identify social norms (\citealp{Bertrand_etal:2015}), study a tax-evasion behavior of taxpayers (\citealp{Breunig_etal:2024}), detect $p$-hacking (\citealp{Elliott_etal:2022}), and test the score manipulation in regression discontinuity (RD) analysis (see \citealp{Cattaneo_Titiunik:2022, Cattaneo_etal:2023med}).

For such purposes, where the value of the density at the boundary is of interest, one of the most widely used approaches is the local polynomial density (LPD) estimator proposed by \cite{Cattaneo_etal:2020}.
The LPD estimator is based on the local polynomial kernel smoothing and, as the authors note, it ``enjoys all the desirable features associated with local polynomial regression estimation," including boundary adaptation.
Due to this superior bias property at boundary points, the LPD estimator has been widely adopted in empirical economics (e.g., \citealp{Britto_etal:ECTA, Chen_etal:JDevEcon, Connolly_Haeck:JLaborEcon, Dasgupta_etal:REStat, Gorrin_etal:2023JIE, He_etal:2020, Khanna:JPE}, among others).

However, despite its popularity, our numerical simulations reveal several variance-related issues with the LPD at the boundary.
First, under the commonly used triangular kernel, confidence intervals (CIs) are often wide near the boundary, and discontinuity tests can be so underpowered that they frequently fail to detect discontinuities, even in large samples.
While this issue is not explicitly discussed in the literature, similar patterns of wide CIs appear in several empirical studies, e.g., \citet[Figure 5]{Breunig_etal:2024}, \citet[Figure 1]{DeBenedetto_etal:2025}, \citet[Figure 5]{Forderer_Burtch:2025}, and \citet[Figure 3]{Keefer_Vlaicu:2025}. 
This feature may limit our ability to draw reliable conclusions about the boundary behavior of the density function or to detect economically important empirical observations, such as bunching.

Furthermore, in addition to this large-sample variance issue, our simulation also indicates that the variance of the LPD estimator under the triangular kernel effectively \textit{diverges} when the sample size is small. This finite-sample variance property further raises concerns about the credibility of boundary estimation and inference based on the LPD estimator.

Interestingly, our simulations suggest that these problematic properties of the LPD estimator are substantially mitigated by using Gaussian or Laplace kernels, although these choices are uncommon in the local polynomial smoothing literature.
In particular, relative to the commonly used triangular kernel, switching to these less common kernels reduces the mean squared error (MSE), shortens confidence intervals, increases statistical power in discontinuity detection, and suppresses the finite-sample variance explosion.
Our first contribution is to document this practically important observation for applied work.

Motivated by these numerical observations, this article theoretically studies and highlights a close connection between kernel choice and these properties of large- and small-sample variance.
First, in sharp contrast to conventional wisdom in kernel estimation literature, we argue that kernel selection has a nontrivial effect on the LPD’s asymptotic efficiency.
For example, the commonly used triangular kernel is 14\% less efficient compared to our preferred Laplace kernel in terms of MSE, a sizable gap given that replacing the optimal Epanechnikov kernel with, for example, a Gaussian kernel reduces efficiency by only about 4\% in the standard setting (\citealp[p.~341]{Hansen:2022_prob}). 
Moreover, for inference, the efficiency loss incurred by using the triangular kernel instead of the Laplace amounts to about 50\%. 
Furthermore, in contrast to the standard local polynomial regression literature \citep{calonico2022coverage}, the uniform kernel is an even \textit{worse} option; its efficiency loss amounts to approximately 75\%.
Intuitively, these values imply that achieving the same interval length as with the Laplace kernel requires approximately 1.6-1.9 times the sample size.
Building on this analysis, we further study the statistical power of the LPD-based discontinuity test against a $\sqrt{nh}$ local alternative.
In contrast to the fixed alternative case studied by \cite{Cattaneo_etal:2020}, the kernel choice has a first-order contribution to the asymptotic power property, and we find that the Gaussian and Laplace kernels can improve the power over the commonly employed kernel functions.

Why does the performance of the LPD estimator depend so strongly on the choice of kernel?
To build intuition for these efficiency gains, we study the equivalent kernels of the LPD estimator. By deriving the infeasible optimal weighting function and comparing it with the equivalent kernels, we find that popular kernels deviate substantially from the optimum, thereby implying sizable efficiency gains from employing alternative kernel functions.
Moreover, this analysis indicates that the gains are rooted in the boundary-kernel methods of \cite{Gasser_etal:1985} and \cite{Muller:1991}, although the numerical importance of kernel choice has been largely overlooked in the literature. Taken together, our results provide new insights for the standard kernel-smoothing literature.

Regarding the small-sample variance inflation issue, we formally show that the LPD estimator inherits not only the ``desirable features" of local polynomial smoothing techniques but also an \textit{undesirable} finite-sample property: when a compactly supported kernel is used, the estimator has no finite variance \citep{Seifert_Gasser:1996}. 
In particular, we show that the finite-sample variance is infinite when using compactly supported kernels such as the triangular and uniform kernels, whereas it is bounded when employing an unbounded support kernel such as the Laplace kernel. 

This small-sample variance property is particularly consequential in score manipulation tests in RD designs, where one-sided manipulation always reduces the effective sample size on one side of the cutoff, potentially making the finite-sample variance problem dominant. 
In this view, the use of an unbounded support kernel is particularly recommended in manipulation testing settings---one of the most widespread applications of LPD estimators.

Taken all together, our results show that \textit{kernel choice matters} both asymptotically and in finite samples.
Although it is commonly believed that, for kernel-based estimators, ``the choice of kernel function $K$ is not very important for the performance of the resulting estimators, both theoretically and empirically" (\citealp[p.~76]{Fan_Gijbels:1996}), the findings in this paper suggest that this (generally valid) understanding does not apply to the boundary estimation and inference using the LPD estimators.
Careful selection of the kernel function can substantially enhance the performance of LPD estimation and inference.
This simple yet powerful modification yields significant efficiency gains: relative to commonly used kernels, the improvements are considerable both theoretically and empirically.

\paragraph{Plan of the Article.}
This paper proceeds as follows. 
Section \ref{sec: motivation} presents numerical and empirical examples that motivate our analysis. 
Section \ref{sec: efficiency} studies the asymptotic efficiency of the LPD estimator in both estimation and inference settings, and then examines statistical power under local alternatives. We also provide intuition for the efficiency gains from kernel choice through an equivalent-kernel analysis. 
Section \ref{sec: finite} investigates the finite-sample variance properties of the LPD estimator. 
Section \ref{sec: recommendation} summarizes practical recommendations for empirical researchers. 
Section \ref{sec: conclusion} concludes. 
All proofs, together with additional discussion of interior points, are provided in the Online Appendix.
The links to the kernel estimation literature are made explicit in Sections \ref{subsec: asy efficiency} and \ref{subsec: why kernel choice matters}, and hence, to conserve space, we refrain from providing a separate literature review.

\section{Motivating Numerical Evidence}\label{sec: motivation}
Kernel choice is often regarded as less important.
To illustrate the potential importance of kernel selection in the LPD estimation and inference, we begin by presenting a piece of numerical evidence. 
\subsection{Discontinuity Detection}\label{subsec: simu discontinuity}
\subsubsection{Simulation Evidence}
We begin with one of the most relevant scenarios, discontinuity detection (or manipulation testing) in RD designs.
We consider the following data-generating process, which is designed to mimic a one-sided score manipulation.
\begin{itemize}
    \item[1.] Randomly draw $n$ observations from a $\mathcal{N}(0,1)$ distribution truncated below at $0$.
    \item[2.] Randomly select $q\times 100\%$ of the observations that fall in $(c, 0.9)$, and replace them with $0.9 + u_i$, where $u_i \sim \mathrm{Uniform}(0, 0.2)$.
\end{itemize}
That is, we consider a scenario where some observations that would otherwise fall in $(c, 0.9)$ are shifted to $[0.9, 0.9 + 0.2]$, with $0.9$ being the cutoff.
We examine the following three cases: (1) $q = 0.5$ and $c = 0.8$ (2) $q = 0.5$ and $c = 0.7$, and (3) the no-manipulation case, i.e., $q=0$.
A smaller value of $c$ (i.e., case (2)) implies that the manipulation occurs over a wider region. For cases (1) and (2), see Figure \ref{fig:hist} for example histograms.
\begin{figure}[t]
    \begin{center}
    \begin{tabular}{cc}
          \begin{minipage}[t]{0.45\hsize}
            \centering
            \includegraphics[keepaspectratio, scale=0.6]{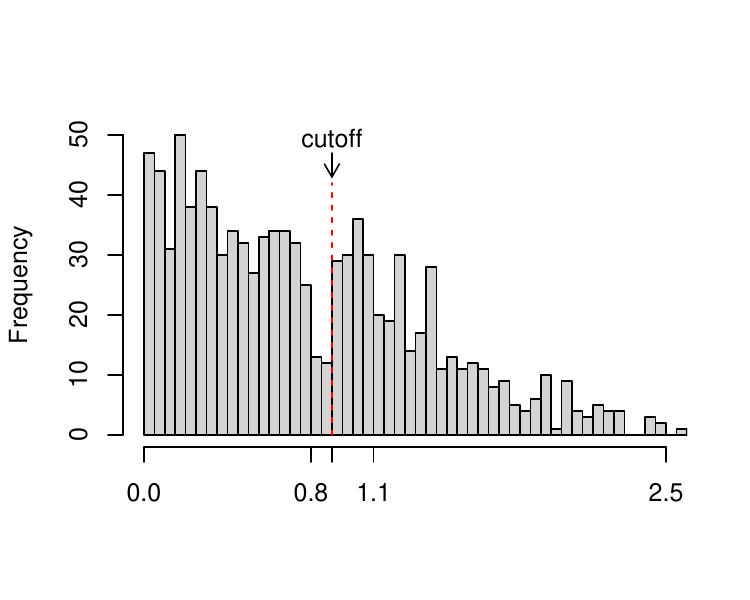}
            \subcaption{Case (1)}
          \end{minipage} &
          \begin{minipage}[t]{0.45\hsize}
            \centering
            \includegraphics[keepaspectratio, scale=0.6]{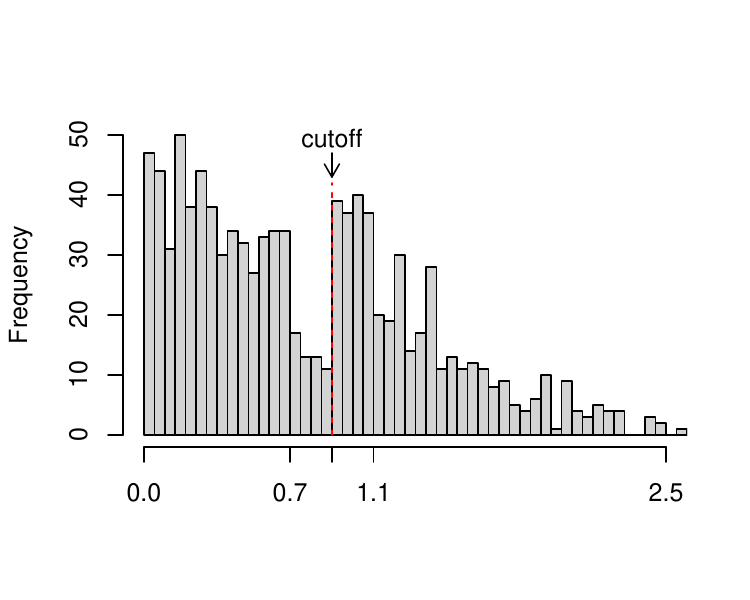}
            \subcaption{Case (2)}
          \end{minipage} 
        \end{tabular}
    \end{center}
    \caption{Histogram Examples from the Simulation}
    \label{fig:hist}
\end{figure}

We perform the discontinuity test by following the procedure proposed in \cite{Cattaneo_etal:2020}, details of which will be introduced in a later section.
The sample sizes are set to $n = 1000$, $750$, and $500$.
Each experiment is repeated 2000 times.
The nominal significance level is set to $0.05$.
\begin{table}[t]
    \begin{center}
    \begin{tabular}{l|ccc|ccc|ccc}
    \hline\hline
         & \multicolumn{3}{c|}{$n=1000$} & \multicolumn{3}{c|}{$n=750$} & \multicolumn{3}{c}{$n=500$} \\
        Kernel & Tri & Gau & Lap & Tri & Gau & Lap & Tri & Gau & Lap\\ \hline
        (1) & 75.95 & 87.75 & 96.65 & 64.65 & 77.65 & 89.35 & 47.40 & 56.45 & 71.40\\
        (2) & 53.00 & 77.50 & 94.25 & 44.95 & 72.45 & 92.85 & 37.60 & 65.50 & 87.15\\
        (3) & 4.35 & 4.70 & 5.25 & 4.20 & 4.45 & 4.60 & 3.80 & 4.55 & 4.95\\\hline
    \end{tabular}
    \caption{Rejection Probabilities of the Null of Density Continuity (\%)}
    \label{tab: rej probs}
    \end{center}

    \footnotesize
    \renewcommand{\baselineskip}{11pt}
    \textbf{Note:} We use the \texttt{rddensity} package with default options for the triangular kernel.
    We add manually the Gaussian and Laplace kernels as they are not originally included in the package.
    For the Gaussian and Laplace kernels, the options controlling the minimum number of observations within the bandwidth (\texttt{nLocalMin} and \texttt{nUniqueMin}) are set to zero, since these kernels have an unbounded support and therefore the size of the sample falling within the bandwidth is not informative.
\end{table}

The simulation results are reported in Table \ref{tab: rej probs}.
We find that the empirical power significantly differs depending on the kernel choice.
In every scenario, the Gaussian and Laplace kernels show higher rejection probabilities than the triangular kernel, the most popular choice in empirical studies.
Moreover, the difference in power is numerically significant. In several cases, the power under the Laplace kernel is more than doubled compared to the triangular kernel.
Even with a relatively small sample size, the Laplace and Gaussian kernels continue to yield a satisfactory rejection rate, whereas the triangular kernel exhibits a marked deterioration in performance.
Furthermore, in case (3) where no manipulation occurs, the rejection probability is less sensitive to the kernel selection and it is close to the nominal level in every case.
These simulation results suggest that the kernel choice is consequential in discontinuity detection using the LPD estimator.

\subsubsection{Real Data Example}\label{subsec: simu empirical}
In the previous simulation, we observed that statistical power varies with the kernel choice.
This result suggests that the efficiency property of the LPD estimator may crucially depend on the kernel selection.
We here provide an empirical example supporting this conjecture.

Our first application adopts data from \citet[EEL]{Eggers_etal:2021JME}.
This study employs RD designs and examines the presence of score manipulation using the LPD-based discontinuity test. 
Figure \ref{fig:empirical_a} replicates \citet[Figure 2]{Eggers_etal:2021JME}.
We can see that the CIs around the cutoff point are excessively wide, particularly on the right-hand side.
This raises concerns about whether a ``no-discontinuity" conclusion can be drawn with confidence, even if the test does not reject the null hypothesis of continuity. 
By contrast, the Gaussian and Laplace kernels (Figures \ref{fig:empirical_b}-\ref{fig:empirical_c}) show much tighter CIs. For example, the CIs under the Laplace kernel are $16\%$ shorter on the left and $41\%$ shorter on the right.
Although the CIs are much tighter, they still exhibit overlap around the cutoff point (and the formal test does not reject the continuity). This persistent overlap, despite the narrower intervals, suggests that manipulation is unlikely and thereby strengthens our confidence in the validity of the main RD study.

\begin{figure}[t]
  \begin{center}
    \begin{tabular}{ccc}
      \begin{minipage}[t]{0.32\hsize}
        \centering
        \includegraphics[keepaspectratio, scale=0.37]{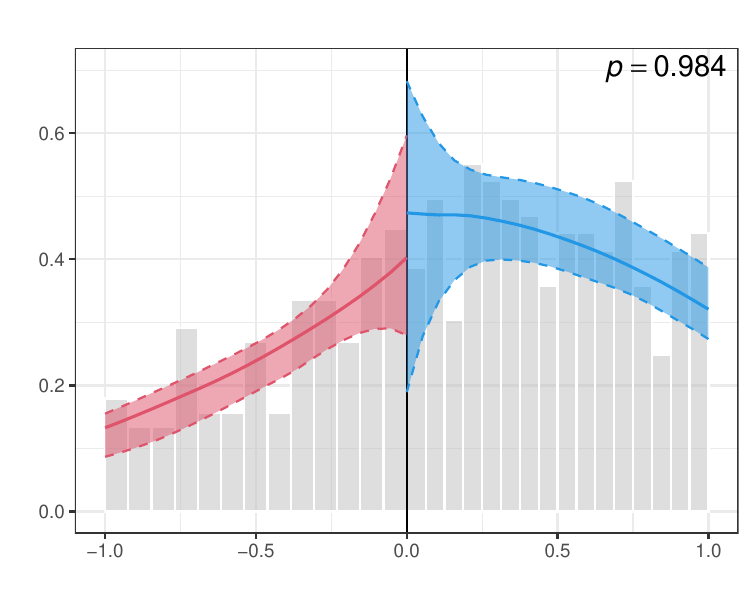}
        \subcaption{EEL: Triangular Kernel}
        \label{fig:empirical_a}
      \end{minipage} &
      \begin{minipage}[t]{0.32\hsize}
        \centering
        \includegraphics[keepaspectratio, scale=0.37]{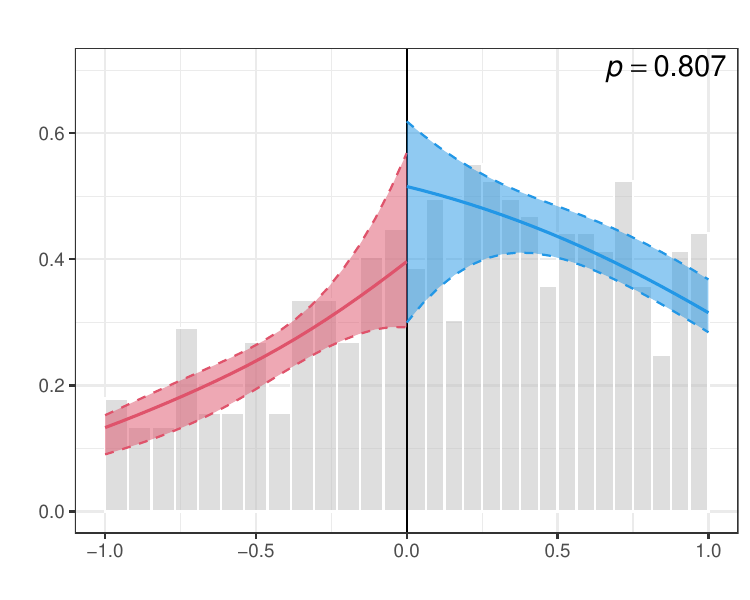}
        \subcaption{EEL: Gaussian Kernel}
        \label{fig:empirical_b}
      \end{minipage} &
      \begin{minipage}[t]{0.32\hsize}
        \centering
        \includegraphics[keepaspectratio, scale=0.37]{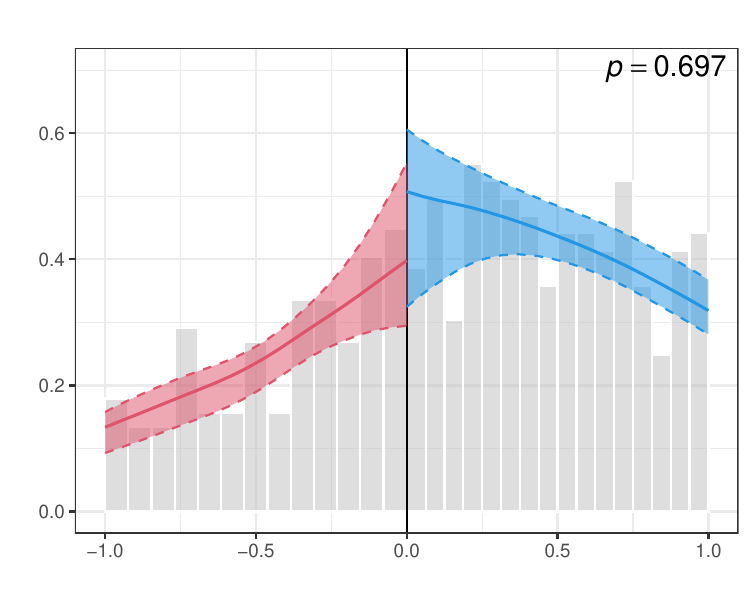}
        \subcaption{EEL: Laplace Kernel}
        \label{fig:empirical_c}
      \end{minipage} \\[1ex]

      \begin{minipage}[t]{0.32\hsize}
        \centering
        \includegraphics[keepaspectratio, scale=0.37]{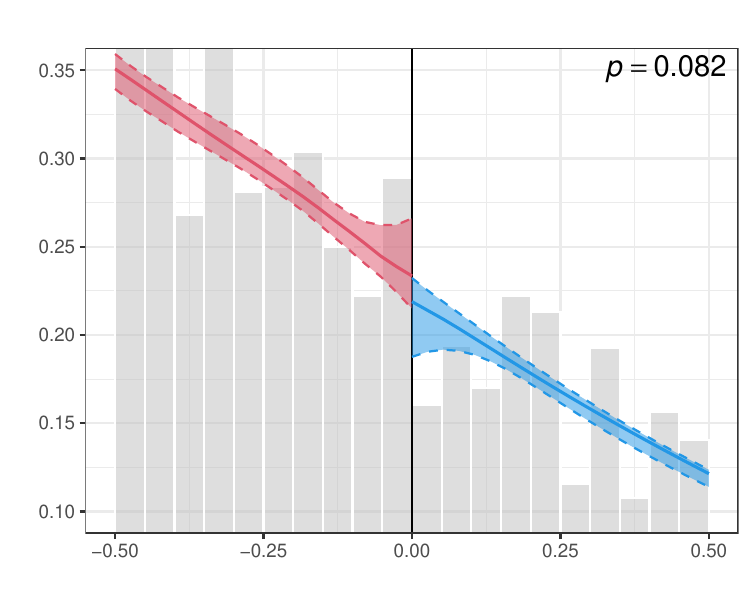}
        \subcaption{LSO: Triangular Kernel}
        \label{fig:empirical_d}
      \end{minipage} &
      \begin{minipage}[t]{0.32\hsize}
        \centering
        \includegraphics[keepaspectratio, scale=0.37]{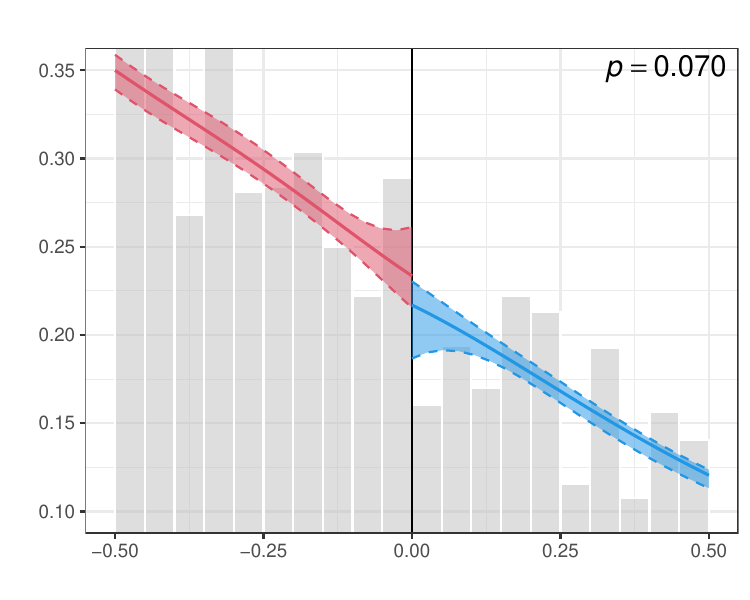}
        \subcaption{LSO: Gaussian Kernel}
        \label{fig:empirical_e}
      \end{minipage} &
      \begin{minipage}[t]{0.32\hsize}
        \centering
        \includegraphics[keepaspectratio, scale=0.37]{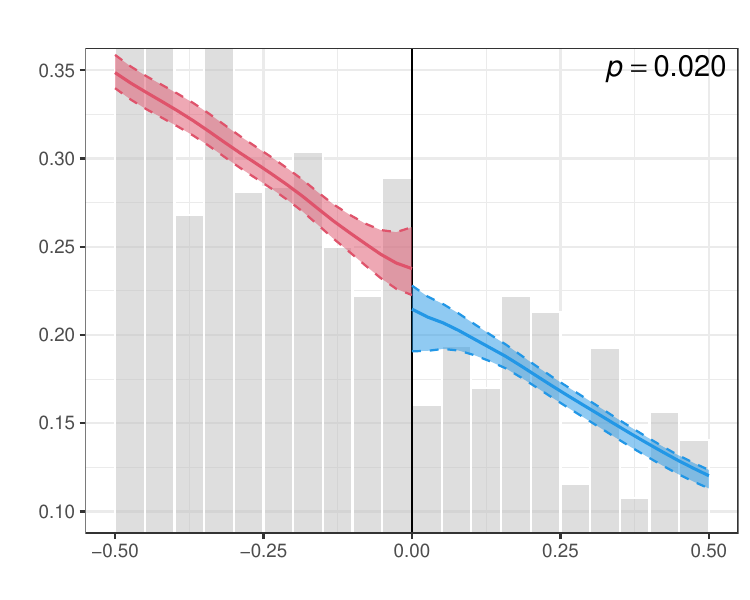}
        \subcaption{LSO: Laplace Kernel}
        \label{fig:empirical_f}
      \end{minipage}
    \end{tabular}
  \end{center}
  \caption{Empirical Examples: Manipulation Testing}
  \label{fig:empirical}

  \footnotesize
  \renewcommand{\baselineskip}{11pt}
  \textbf{Note:} The LPD estimates (solid lines) and 95\% pointwise confidence intervals (shaded areas) are shown. The sample sizes are $n=580$ in the upper panel and $n=40582$ in the lower panel.
\end{figure}

As an additional empirical example, we use data from \citet[LSO]{Lindo_etal:2010}, which is also an RD study.
Figure \ref{fig:empirical_d} reports the density estimates together with CIs, where the CIs exhibit moderate overlap. Consistent with this, the discontinuity test yields a $p$-value of $0.082$, so we fail to reject continuity at the conventional nominal level. 
Under the Gaussian kernel, the conclusion is unchanged, while the CIs become slightly shorter.
Using the Laplace kernel, the CI length decreases further (by $24\%$ on the left and $16\%$ on the right compared to the triangular kernel) with limited overlap. 
Consistently, the discontinuity test rejects continuity, with a $p$-value of $0.02$. This discontinuity detected under the Laplace kernel is consistent with the binomial test performed in \citet[Chapter 4.4]{cattaneo2023practical}, which is based on the local randomization framework proposed in \cite{Cattaneo_etal:2015}.

\subsection{Density Estimation}\label{subsec: simu density}
To see the variance property more directly, we next consider the standard density estimation setting at a boundary point. 
The data-generating process is the standard normal distribution $\mathcal{N}(0,1)$ truncated below at $-0.8$, which is adopted from \citet[Supplemental Appendix]{Cattaneo_etal:2020}. 
We are interested in $f(-0.8)$.
We perform estimation and inference using the \texttt{lpdensity} package with default options.

\begin{table}[!t]
    \begin{center}
    \begin{tabular}{l ccc c ccc}
        \hline\hline
         & \multicolumn{3}{c}{Estimated Bandwidth} && \multicolumn{3}{c}{Theoretical Bandwidth} \\
        \cline{2-4} \cline{6-8}
        Kernel & Tri & Gau & Lap && Tri & Gau & Lap \\
        \hline
        \multicolumn{6}{l}{($n=1000$)}\\
        Coverage & 0.960 & 0.950  & 0.946  && 0.947 & 0.947 & 0.947  \\
        Length    & 0.385 & 0.314  &  0.286 && 0.331 & 0.301  & 0.267 \\
        MSE ($\times10^2$)   &  0.475 &  0.473 & 0.386  &&  0.299 & 0.280 & 0.257
        \\
        \hline
        \multicolumn{6}{l}{($n=750$)}\\
        Coverage & 0.960 & 0.946  & 0.946  && 0.951 &  0.948 & 0.948  \\
        Length    & 0.440 & 0.362 &  0.330 && 0.376 & 0.341  & 0.302 \\
        MSE ($\times10^2$)    &  0.622 & 0.558 & 0.453 &&  0.382 & 0.355 & 0.325\\
        \hline
        \multicolumn{6}{l}{($n=500$)}\\
        Coverage & 0.955 &  0.940  & 0.941  && 0.950 &  0.951  & 0.945  \\
        Length    & 0.533 & 0.442  &  0.404 && 0.449 &  0.407 & 0.362  \\
        MSE ($\times10^2$)    & 0.844 &  0.728 & 0.636  && 0.559 & 0.520  & 0.479\\
        \hline
    \end{tabular}
    \caption{Coverage Probability, Average Length, and MSE}
    \label{tab: intro simulation}
    \end{center}

    \footnotesize
    \renewcommand{\baselineskip}{11pt}
    \textbf{Note:} The empirical coverage probability, average length of the pointwise CI, and MSE at the left boundary after $2000$ iterations are shown. 
    On the left panel, the estimated bandwidths using the \texttt{lpdensity} package are used. On the right panel, the theoretical asymptotically optimal bandwidths are used (see Section \ref{sec: inference} for further details on the optimal bandwidth). 
\end{table}
Table \ref{tab: intro simulation} reports the empirical coverage probabilities, average CI lengths, and MSEs based on 2000 iterations.
Although the coverage probabilities are nearly identical across kernels and close to the nominal level, the average CI length with the Laplace kernel is substantially shorter than with the triangular kernel.
In every case, the Laplace kernel yields CIs about 25\% shorter than those of the triangular kernel. The Gaussian kernel also shows similar improvement, although modest compared to the Laplace.
The MSEs follow a similar pattern: The Gaussian and Laplace kernels attain an MSE approximately 10-25\% lower than that of the triangular kernel.

To eliminate the influence of the variation in the selected bandwidth, Table \ref{tab: intro simulation} also reports the performance based on the (infeasible) asymptotically optimal bandwidths at every iteration. Once again, we observe the same pattern. The average length and MSE under the Laplace or Gaussian kernels are much shorter/smaller than under the triangular kernel.

\subsection{Small Sample Performance}\label{subsec: simu small}
Finally, we provide simulation evidence on the small-sample behavior of the LPD estimator.
We again use the truncated normal distribution as in Section \ref{subsec: simu density}, but now consider cases where the sample size is much smaller ($n\in \{400, 350, \ldots, 100, 75, 50\}$).
The empirical relevance of such small-sample settings will be discussed in Section \ref{sec finite manip}; here, we simply report and interpret the simulation results.
\begin{figure}[t]
    \begin{center}
        \includegraphics[width=0.5\linewidth]{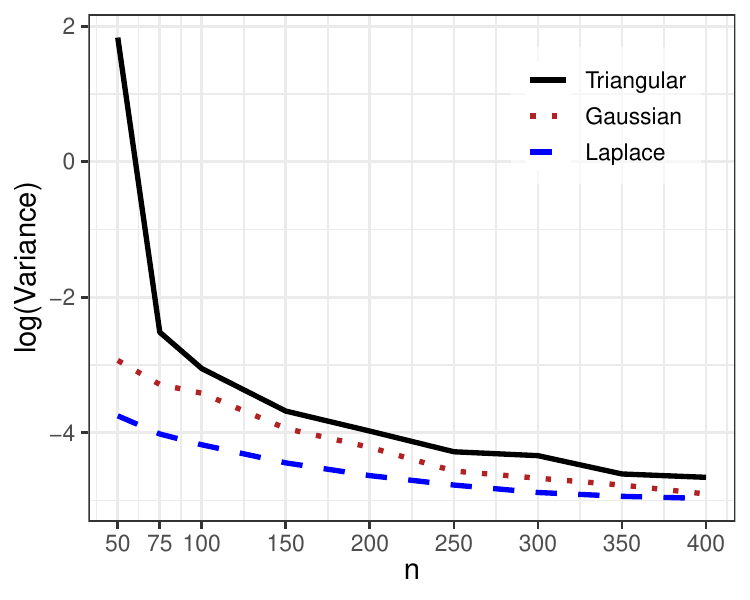}
    \caption{Small-Sample Variance}
     \label{fig: finite variance}
    \end{center}

     \footnotesize
    \renewcommand{\baselineskip}{11pt}
    \textbf{Note:} The variance is reported on the logarithmic scale. 
    The bandwidths are estimated using the \texttt{lpdensity} package. To avoid stabilizing the estimation in an artificial way that could obscure the underlying problem, the options controlling the minimum number of observations within the bandwidth (\texttt{nLocalMin} and \texttt{nUniqueMin}) are set to zero. 
    As a result, when using the triangular kernel, the LPD estimator occasionally fails to produce estimates when the sample size is small (1 time out of 1000 iterations when $n=100$ and $n=75$, and 20 times when $n=50$).
    The variance is computed after excluding such cases.
    In contrast, the Gaussian and Laplace kernel successfully produces estimates in every case.
\end{figure}

Figure \ref{fig: finite variance} reports the variance (in the logarithmic scale) of the LPD estimator based on $1000$ iterations for each sample size.
We find that the finite-sample variance exhibits a very different pattern depending on the kernel function used.
When the popular triangular kernel is used, the small-sample variance increases explosively as the sample size decreases. 
In contrast, with the Gaussian or Laplace kernels, the variance inflation is suppressed.\footnote{A similar pattern is observed for the MSE, since the variance becomes so large that it dominates the squared bias (see the Online Appendix, Section S4.1).}
Therefore, when the sample size is small---or more broadly, when the local sample size near the evaluation point is limited---the kernel selection can substantially affect the stability of the LPD estimator.

\subsection{Summary of the Numerical Findings}
In Sections \ref{subsec: simu discontinuity} and \ref{subsec: simu density}, we observed that the efficiency property depends on the kernel function used, even for large sample sizes. 
The statistical power was insufficient to detect the discontinuity when using the triangular kernel, the most common choice.
The length of the CI was also so large under the triangular kernel that the interpretation of the statistical inference can be unclear. 
MSE suggested a similar conclusion.
In contrast, the LPD estimator using the Gaussian or Laplace kernels performs better.
The length of the CI shrinks, and the power property in discontinuity testing improves.
MSE is also smaller under the Gaussian or Laplace kernels than the triangular kernel.

Similarly, Section \ref{subsec: simu small} suggests that the LPD estimator using the triangular kernel is still inferior when the sample size is small.
In particular, the finite-sample variance explodes quickly, whereas this explosion is suppressed if we employ the Gaussian or Laplace kernels instead. 

In the following sections, we formally investigate the root cause of why the efficiency properties of the LPD estimator depend on the choice of kernel.
The analysis in Section \ref{sec: efficiency} corresponds to Sections \ref{subsec: simu discontinuity} and \ref{subsec: simu density} in that we investigate the asymptotic properties of the LPD estimator. 
After that, in Section \ref{sec: finite}, we theoretically explain the observation in Section \ref{subsec: simu small} by establishing a finite-sample theory that does not rely on asymptotic approximations.

\section{Kernel Selection and Large Sample Properties}\label{sec: efficiency}
In this section, we examine how the asymptotic performance of the LPD estimator depends on the choice of kernel function.
The next subsection introduces notation and several preliminary results, after which we study the asymptotic efficiency among two classes of kernel functions in Section \ref{subsec: asy efficiency}.
We then investigate the statistical power property in Section \ref{subsec: power}.
After that, Section \ref{subsec: why kernel choice matters} explains why this dependence arises, using the equivalent kernel framework.

\subsection{Preliminaries}
Let $\mathcal{X} = (-\infty,x_R) \subset \mathbb{R}$ be the data domain. 
We are given an independent and identically distributed sample $X_{1:n}  = (X_1,\ldots,X_n)$ defined on $\mathcal{X}$ with unknown distribution $F$ with density $f$.
We are interested in the density at the boundary, $f(x_R)$. Throughout this article, we assume $f(x_R) >0$.
The LPD estimator of \cite{Cattaneo_etal:2020} at the boundary point $x_R$ is given by $\hat{f}(x_R)= \bm{e}_1^\prime \hat{\bm{\beta}}(x_R)$, where
\begin{align}
    \hat{\bm{\beta}}(x_R) = \argmin_{\bm{\beta}\in\mathbb{R}^{p+1}} \sum_{i=1}^{n}\left\{
    \hat{F}(X_i) - \bm{r}_p(X_i - x_R)^\prime \bm{\beta} 
    \right\}^2
    \frac{1}{h}K\left(\frac{X_i - x_R}{h}\right),\label{objective function}
\end{align}
where $\bm{e}_1 = (0,1,0,\ldots,0)^\prime$; $\hat{F}(x)=n^{-1}\sum_{i=1}^{n}\mathbf{1}\left\{X_i \leq x\right\}$; $\bm{r}_p(u)=(1, u, u^2, \ldots, u^p)^\prime$; $K(\cdot)$ is a non-negative, symmetric kernel function such that $\int K(u)du=1$; $h(>0)$ is a bandwidth; and $p(\geq1)$ is the polynomial degree.

The asymptotic bias and variance of $\hat{f}(x_R)$ are obtained under standard assumptions in \citet[p.~1451]{Cattaneo_etal:2020} as
\begin{align*}
    & \mathrm{Bias}\left[ \hat{f}(x_R) \right] \approx  \frac{f^{(p)}(x_R)}{(p+1)!} \times \mathcal{B}_{p,K} \times h^p, 
    \quad 
    \mathbb{V}\left[ \hat{f}(x_R) \right] \approx \frac{f(x_R)}{nh}\times\mathcal{V}_{p,K},
\end{align*}
respectively, where $\mathcal{B}_{p,K} = \bm{e}_1^\prime \bm{A}_{p,K}^{-1} \bm{c}_{p,K}$, and $\mathcal{V}_{p,K} = \bm{e}_1^\prime \bm{A}_{p,K}^{-1}\bm{B}_{p,K} \bm{A}_{p,K}^{-1} \bm{e}_1$, with 
\begin{align}
    \bm{A}_{p,K} &= \int_{-\infty}^{0} \bm{r}_p(u)\bm{r}_p(u)^\prime K(u)\, du,\quad
    \bm{c}_{p,K} = \int_{-\infty}^0 \bm{r}_p(u)u^{p+1} K(u)\, du,\notag\\
    \bm{B}_{p,K} &= \int_{-\infty}^{0}\int_{-\infty}^{0}
    \min\{u,v\}\bm{r}_p(u)\bm{r}_p(v)^\prime K(u)K(v) \,dudv\label{Bp}. 
\end{align}
Then, the asymptotic MSE optimal bandwidth can be deduced from these as $h_p^{\mathtt{MSE}} = (
{\mathcal{V}_{p,K}}/\{C_1(p,x_R,F) \mathcal{B}_{p,K}^2\})^{1/(2p+1)} n^{-1/(2p+1)}$, where $C_1(p,x_R,F)$ is some constant that depends only on $p$, $x_R$, and $F$ \citep[p.~1452]{Cattaneo_etal:2020}.
Using these results, the MSE at $h_p^{\mathtt{MSE}}$ is asymptotically given by
\begin{align*}
    \mathrm{MSE}\left[ \hat{f}(x_R) \,;\, h_p^{\mathtt{MSE}} \right] 
    \approx C_2(p,x_R,F) \times \left(\mathcal{B}_{p,K}^2\right)^{1/(2p+1)} \mathcal{V}_{p,K}^{2p/(2p+1)} \times  n^{-2p/(2p+1)}
\end{align*}
with some constant $C_2(p,x_R,F)$ that does not depend on $K$.
We will write the constant part depending on the kernel function as $\mathcal{Q}_{p, K} = (\mathcal{B}_{p,K}^2)^{1/(2p+1)} \mathcal{V}_{p,K}^{2p/(2p+1)}$.

\subsection{Asymptotic Efficiency}\label{subsec: asy efficiency}
\subsubsection{Point Estimation}
We begin by examining how the asymptotic variance and MSE depend on the choice of kernel function.
Based on the preliminary results in the previous subsection, we can analyze this dependence through $\mathcal{V}_{p,K}$ and $\mathcal{Q}_{p,K}$, which characterize the contribution of the kernel function to the variance and MSE-efficiency properties of the LPD estimator.\footnote{The usefulness of the indicator $\mathcal{V}_{p,K}$ may be debatable, as it compares the constant term of the asymptotic variance under a fixed bandwidth, which may not reflect recent empirical practice. Even for inference, bandwidth selection rules typically account for bias in some way, so quantities such as $\Theta_K$ (defined later) are arguably more directly relevant empirically. Nevertheless, we report $\mathcal{V}_{p,K}$ here for completeness and to maintain alignment with the classical kernel smoothing literature. For example, \cite{Gasser_etal:1985} defines and derives the minimum-variance kernel that minimizes the counterpart of this quantity in standard settings.}

We compare two classes of kernel functions.
Following \cite{Fan_Gijbels:1995, Fan_Gijbels:1996}, our first candidate is the Beta family: $\{(1-u^2)_{+}\}^\gamma/\mathrm{Beta}(1/2, \gamma+1)$. Among this family, we consider the uniform ($\gamma=0$), the Epanechnikov ($\gamma=1$), the biweight ($\gamma=2$), and the Gaussian kernels ($\gamma\to\infty$ with appropriate scaling).
In addition to these kernels, to study the case when the kernels are more centered, we also consider a generalized triangular family: $\{(1-|u|)_+\}^m (m+1)/2$.
Among this family, we consider the triangular ($m=1$), what we call $2$-triangular ($m=2$), $3$-triangular ($m=3$), and the Laplace kernels ($m\to\infty$). Note that $m=0$ corresponds to the uniform kernel.
Together, these two classes cover most of the popular kernel functions used in the kernel smoothing literature.

Our kernel selection is also motivated as follows. 
First, the uniform, triangular, and Epanechnikov kernels are implemented in the software packages provided by \citet{Cattaneo_etal:2018, Cattaneo_etal:2022}, and are therefore natural benchmarks. 
In addition, in standard local polynomial regression, the triangular kernel is the MSE-optimal choice at a boundary point \citep{Cheng_etal:1997}, while the uniform kernel is a popular choice for inference \citep{calonico2022coverage} and is the MSE-optimal choice in the LPD estimation at interior points \citep{cattaneo2021local}.
The Epanechnikov kernel is the well-known optimal kernel in the standard kernel estimation at interior points (\citealp{Epanechnikov:1969}), while the biweight kernel is also optimal among kernels with a certain smoothness (\citealp{Muller:1984}).
Furthermore, the Gaussian kernel is a standard choice in the kernel estimation literature and has several optimality properties (\citealp{Cline:1988, Granovsky_Muller:1991}). 
The Laplace kernel is the optimal choice for non-smooth densities in standard kernel estimation \citep{vanEeden:1985}.
Although the $2$- and $3$-triangular kernels seem less standard in the literature, we include them as they are a natural intermediate between the triangular and Laplace kernels.

The asymptotic variance ($\mathcal{V}_{p,K}$) and MSE-efficiency ($\mathcal{Q}_{p,K}$) for these kernels and $p=1,2,3$ are summarized in Table \ref{tab: kernel asy efficiency}.
\begin{table}[t]
    \begin{center}
    \scalebox{1}{
    \begin{tabular}{l|ccc|ccc|c}
    \hline\hline
        \multirow{2}*{Kernel Function} & \multicolumn{3}{c|}{$\mathcal{V}_{p,K}$} & \multicolumn{3}{c|}{$\mathcal{Q}_{p,K}$} & \multirow{2}{*}{$\Theta_K$}\\
         & $\mathcal{V}_{1,K}$ & $\mathcal{V}_{2,K}$ & $\mathcal{V}_{3,K}$ & $\mathcal{Q}_{1,K}$ & $\mathcal{Q}_{2,K}$ & $\mathcal{Q}_{3,K}$\\ \hline
        \multirow{2}*{Uniform} & 1.200 & 5.486 & 14.286 & 1.129 & 3.182 & 6.831 & 8.285\\ 
         & (4.80) & (7.31) & (10.16) & (1.13) & (1.26) & (1.38)& (1.75)\\ \hline
        \multirow{2}*{Epanechnikov} & 1.345 & 5.772 & 14.468 & 1.087 & 2.991 &6.326 & 7.496\\ 
         & (5.38) & (7.70) & (10.29) & (1.09) & (1.18) & (1.28)& (1.58)\\ \hline
         \multirow{2}*{Biweight} & 1.496 & 6.184 & 15.108 & 1.071 & 2.912 & 6.103  & 7.114\\ 
         & (5.99) & (8.24) & (10.74) & (1.07) & (1.15) & (1.24) & (1.50)\\ \hline
         \multirow{2}*{Gaussian} &  0.484 & 1.748 & 3.809 & 1.041 & 2.739 & 5.566 & 5.967\\
         & (1.93) & (2.33) & (2.71) & (1.04) & (1.09) & (1.13)& (1.26)\\ \hline
         \multirow{2}*{Triangular} & 1.371 & 5.714 & 14.026 & 1.064 & 2.873 & 5.989 & 7.053\\ 
         & (5.49) & (7.62) & (9.97) & (1.06) & (1.14) & (1.21)& (1.49)\\ \hline
         \multirow{2}*{$2$-Triangular} & 1.587 & 6.234 & 14.685 & 1.038 & 2.746 & 5.626 & 6.468\\ 
         & (6.35) & (8.31) & (10.44) & (1.04) & (1.09) & (1.14) & (1.37)\\ \hline
         \multirow{2}*{$3$-Triangular} & 1.818 & 6.853 & 15.664 & 1.026 & 2.678 & 5.428 & 6.122\\ 
         & (7.27) & (9.14) & (11.14) & (1.03) & (1.06) & (1.10) & (1.29)\\ \hline
        \multirow{2}*{Laplace} & 0.250 & 0.750 & 1.406 & 1.000 & 2.524 & 4.935 & 4.733\\ 
        & (1.00) & (1.00) & (1.00) & (1.00) & (1.00) & (1.00) & (1.00)\\ \hline
    \end{tabular}}
    \caption{Asymptotic Variance and Efficiency Relative to the Laplace Kernel}
    \label{tab: kernel asy efficiency}
    \end{center}

    \footnotesize
    \renewcommand{\baselineskip}{11pt}
    \textbf{Note:} The constant part of the asymptotic variance ($\mathcal{V}_{p,K}$) and MSE ($\mathcal{Q}_{p,K}$) are shown. $\Theta_K$, defined in Section \ref{sec: inference}, indicates the asymptotic variance under the simple robust bias correction.
    The numbers in parentheses indicate the relative efficiency compared to the Laplace kernel.
\end{table}
We observe that the Laplace kernel not only exhibits better variance properties (at a fixed bandwidth) than commonly used kernels, but also dominates them in terms of asymptotic MSE (at the MSE-optimal bandwidths). 
Notice that, unlike usual, the kernel choice has a \textit{significant} impact on the asymptotic MSE-efficiency. For the case when $p=2$, which is the most standard choice in the literature (\citealp[p.~1452]{Cattaneo_etal:2020}; \citealp[pp.~76--80]{Fan_Gijbels:1996}), the commonly used kernels are 14--26\% less efficient than the Laplace kernel. Intuitively speaking, this implies that the commonly used kernels require approximately 1.2--1.3 times larger sample sizes to achieve the same performance as the Laplace kernel. 

This efficiency loss is in contrast to the standard kernel smoothing at interior points, wherein the efficiency loss between the optimal Epanechnikov kernel and popular choices is quite limited. For example, the efficiency loss from using the Gaussian kernel relative to the optimal Epanechnikov kernel is only $4\%$ (e.g, \citealp[p.~341]{Hansen:2022_prob}).
In this view, contrary to the common belief in kernel smoothing literature, a careful kernel choice is essential in the boundary estimation using the LPD estimator.

\subsubsection{Inference under the Simple Robust Bias Correction}\label{sec: inference}
We then study the implications of kernel choice for inference.
To perform statistical inference, we have to deal with the asymptotic smoothing bias.
\cite{Cattaneo_etal:2020} propose to rely on a simple robust bias correction approach (\citealp[Remark 7]{Calonico_etal:2014}; \citealp{calonico2018effect}). 
Specifically, the authors propose using local cubic estimation (i.e., $p=3$) together with the MSE optimal bandwidth for $p=2$ (i.e., $h_2^{\mathtt{MSE}}$).
By doing so, the (bias-corrected) LPD estimator, $\hat{f}_{\mathrm{bc}}$, is correctly centered.

The variance of $\hat{f}_{\mathrm{bc}}$ is asymptotically approximated by
\begin{align*}
    \mathbb{V}\left[ \hat{f}_{\mathrm{bc}}(x_R) \right]\approx f(x_R)\mathcal{V}_{3,K}/(nh_2^{\mathtt{MSE}}) = \Theta_K \times C_3(x_R, F) n^{-4/5},
\end{align*}
where $\Theta_{K} = (\mathcal{B}_{2,K}^2)^{1/5} \mathcal{V}_{2,K}^{-1/5} \mathcal{V}_{3,K}$.
Therefore, we can examine how the asymptotic variance depends on $K$ via $\Theta_K$, a key quantity that governs the length of the CIs.

The last column of Table \ref{tab: kernel asy efficiency} summarizes $\Theta_K$ for the eight kernels considered before. 
We observe that the Laplace kernel again performs best among the eight kernels.
We also find that the triangular kernel is approximately 50\% less efficient than the Laplace kernel, which is a remarkably large gap. 
In view of the standard local polynomial smoothing literature, one might be tempted to employ the uniform kernel for inference, as it is the CI-length-optimal kernel \citep{calonico2022coverage}. 
However, this turns out to be an even worse choice, as it exhibits 75\% less efficiency.

These magnitudes are noteworthy. 
The triangular and uniform kernels require sample sizes approximately 1.6 to 1.9 times larger to achieve the same performance as the Laplace kernel in terms of asymptotic variance or interval length. 

\subsection{Statistical Power}\label{subsec: power}
Building on the discussion in Section \ref{sec: inference}, we here investigate the statistical power property of the LPD-based discontinuity testing.
\cite{Cattaneo_etal:2020} shows that, under fixed alternatives, the test is consistent (i.e., its power converges to one asymptotically) regardless of the kernel choice.
While this fixed-alternative analysis is valuable for establishing consistency, it is not designed to be informative about power comparisons across kernels or about the finite-sample differences we documented in Section \ref{subsec: simu discontinuity}; see \citet[Chapter 14]{vanderVaart:1998} for a general discussion.
Motivated by this, and to better understand the power behavior in our setting, we study power against a local alternative.

Consider the testing problem $\mathbb{H}_0: f(0+) = f(0-)$, but the following local alternative is correct:
\begin{align*}
    \mathbb{H}_{1,\mathtt{LA}}: \tau_n \coloneqq f(0+) - f(0-) = {c}n^{-2/5},\quad c\in\mathbb{R}\backslash \{0\}.
\end{align*}
Let $\hat{f}_+(0)$ and $\hat{f}_-(0)$ be the LPD estimators only using the data lying over the right side region and the left side region, respectively.
Following the recommendation of \cite{Cattaneo_etal:2020} and empirical practice, we specifically consider the case where $p=3$ together with the MSE-optimal bandwidths for $p=2$ (i.e., $h_{2+}^{\mathtt{MSE}}$ and $h_{2-}^{\mathtt{MSE}}$).
Then $\tau_n$ can be estimated by $\hat{\tau}_n \coloneqq ({n_+}/{n}) \hat{f}_+(0) - ({n_-}/{n}) \hat{f}_-(0)$.
\cite{Cattaneo_etal:2020} proposes the statistic 
\begin{align*}
    T_{3}\coloneqq \frac{\hat{\tau}_n}{\sqrt{\frac{1}{nh_{2+}^{\mathtt{MSE}}} \hat{\mathscr{V}}_{+} + \frac{1}{nh_{2-}^{\mathtt{MSE}}} \hat{\mathscr{V}}_{-}}},
\end{align*}
where $\hat{\mathscr{V}}_{+}$ and $\hat{\mathscr{V}}_{-}$ are consistent estiamtors of the asymptotic variances, ${\mathscr{V}}_{+}\coloneqq f(0+)\mathcal{V}_{3,K}$ and ${\mathscr{V}}_{-}\coloneqq f(0-)\mathcal{V}_{3,K}$.
To study the power against the local alternative $\mathbb{H}_{1,\mathtt{LA}}$, we make the following assumptions. 
Fix $c\in\mathbb{R}\backslash\{0\}$.
\begin{assL} \label{as:DGP-local}
There exists a sequence of distribution functions $\{F_n\}_{n\ge 1}$ such that
\begin{itemize}
    \item[(i)] $\{X_i\}_{i=1}^n$ is an i.i.d.~sample from distribution $F_{n}$ for each $n\in\mathbb{N}_+$;
    \item[(ii)] for all $n\in\mathbb{N}_+$ and $ \ell=0,1,2,3,4$, there exists $C_F<\infty$ such that $|F_{n}^{(\ell)}(x)| <C_F$ separately on $(-s_\ell, 0)$ and $(0,s_r)$, and the density function $f_{n}(x)$ is bounded and bounded away from 0 by some constant independent of $n$ over its support $(-s_\ell,s_r)$, where $s_\ell$ and $s_r$ are some positive constants independent of $n$;
    \item[(iii)] for each $n\in\mathbb{N}_+$, it holds that $f_{n}(0_+) - f_{n}(0_-) = cn^{-2/5}$. 
\end{itemize}
\end{assL}
\begin{assL} \label{as:kernel}
    The kernel function $K$ 
    \begin{itemize}
        \item[(i)] is a nonnegative, symmetric, and continuous second-order kernel;
        \item[(ii)] is either compactly supported or has an exponentially decaying tail. 
    \end{itemize}
\end{assL}
The assumption \ref{as:DGP-local}(ii) is the standard regularity conditions on the data-generating process. 
In particular, this assumption is almost the same as Assumption 2 in the supplemental appendix for \cite{Cattaneo_etal:2020}. 
The only notable difference is that we state it for an $n$-varying sequence $\{F_n\}$ to accommodate the local alternative. 
Plus, we note that we could allow $(s_l,s_r)$ to depend on $n$, while such a generalization would only lengthen the technical discussion and would not yield additional insight for our purposes. 

Under these assumptions and $\mathbb{H}_{1,\mathtt{LA}}$, we have the following distributional approximation.
\begin{lemma}\label{lemma:gaussian approximation}
    Under Assumptions \ref{as:DGP-local} and \ref{as:kernel}, it holds that $\sup_{z\in\mathbb{R}} \left| \mathbb{P}_{n}(T_3 \le z) - \Phi(z-  \tilde{\tau}_n) \right| \to 0$, where $\tilde \tau_n \coloneqq \left( \mathscr{V}_{n,+}/(nh_{2+}^{\mathtt{MSE}}) + \mathscr{V}_{n,-}/(nh_{2-}^{\mathtt{MSE}})\right)^{-1/2}\tau_n$, and $\Phi(\cdot)$ denotes the distribution function of the standard normal distribution $\mathcal{N}(0,1)$.
\end{lemma}
Lemma \ref{lemma:gaussian approximation} justifies the following approximation:
\begin{align*}
    T_3 = \frac{\tau_n}{\sqrt{\frac{1}{nh_{2+}^{\mathtt{MSE}}} {\mathscr{V}}_{+} + \frac{1}{nh_{2-}^{\mathtt{MSE}}} {\mathscr{V}}_{-}}} + Z + o_p(1)
    = \frac{\tau_n}{\sqrt{C_4(F)\Theta_K n^{-4/5}}} + Z + o_p(1),
\end{align*}
where $Z\sim \mathcal{N}(0,1)$, and $C_4(F)>0$ is a constant depending only on the underlying distribution $F$.
We can rewrite as 
\begin{align*}
    T_3 = \frac{\tau_n}{\sqrt{C_4(F)\Theta_K n^{-4/5}}} + Z + o_p(1) = 
    \frac{c/\sqrt{C_4(F)}}{\sqrt{\Theta_K}} + Z + o_p(1).
\end{align*}
Then, the power function can be asymptotically approximated by $\P{|T_3| > z_{1-\alpha/2}} = \pi_\alpha(\Theta_K) + o(1)$, where
\begin{align*}
    \pi_\alpha(\Theta_K)\coloneqq 
    1 - \P{Z \leq z_{1-\alpha/2} - \frac{c/\sqrt{C_4(F)}}{\sqrt{\Theta_K}}} + 
    \P{Z \leq -z_{1-\alpha/2} - \frac{c/\sqrt{C_4(F)}}{\sqrt{\Theta_K}}},
\end{align*}
which is a decreasing function of $\Theta_K$.
Therefore, kernels with a smaller $\Theta_K$ are preferable from the standpoint of statistical power.
Taken together with the previous section's argument that $\Theta_K$ can vary substantially across kernels, this implies that kernel choice can be of first-order importance for power, and that popular kernels can deliver lower power than, for example, the Laplace kernel.

\subsection{Why Does the Kernel Choice Matter?}\label{subsec: why kernel choice matters}
We have observed that kernel selection has a non-trivial impact on the efficiency of the LPD estimator, affecting both the MSE, CI length, and power.
Why does this happen?
This subsection investigates the underlying reason for this dependence.

\subsubsection{MSE}
We first aim to develop intuition for why the MSE varies substantially with the choice of kernel function.
Our discussion relies on the equivalent kernel \citep[p.~63]{Fan_Gijbels:1996}, which is helpful in understanding how the LPD applies weights to each datum point.\footnote{The use of equivalent-kernel representations in the context of LPD estimation is not new and is rather common in the literature. For example, \citet{cattaneo2021local} uses equivalent kernels to study the variance properties of their local regression distribution estimators at interior points and to show that the uniform kernel is MSE-optimal in the LPD estimation at interior points. Our focus, however, is on boundary behavior, and the discussion that follows contains several observations that are novel to both the kernel-smoothing and LPD literatures.}
The LPD estimator can be rewritten as follows (see the Online Appendix Section S2 for details):
\begin{align*}
    \hat{f}(x_R) &= 
    \frac{1}{nh}\sum_{i=1}^{n} 
    \int_{-\infty}^{0} \bm{e}_1^\prime \bm{A}_{p,K}^{-1} \bm{r}_p\left( u\right) K\left( u\right)
    \mathbf{1}\left\{\frac{X_i - x_R}{h} \leq u\right\}\,du + o_p(1).
\end{align*}
This expression suggests that the LPD estimator at boundary points is asymptotically equivalent to the standard kernel density estimator (KDE) using the following (asymmetric) kernel function,
\begin{align}
    K^*_{p,K}(u) = 
    \int_{-\infty}^{0} \bm{e}_1^\prime \bm{A}_{p,K}^{-1} \bm{r}_p\left( z\right) K\left( z\right)
    \mathbf{1}\left\{u \leq z\right\}
    \,dz.\label{equivalent kernel}
\end{align}
Therefore, an investigation of $K^*_{p,K}$ will offer insights into how the LPD utilizes the information contained within the sample.
The following lemma characterizes the basic properties of $K^*_{p,K}$.
\begin{lemma}\label{lemma: equiv kernel}
    Assume $K$ is a symmetric, nonnegative second-order kernel function such that
    \begin{align*}
        \int_{0}^{\infty} u^{2p} K(u) \,du < \infty.
    \end{align*}
    Then $K^*_{p,K}$ defined in \eqref{equivalent kernel} satisfies $K^*_{p,K}(0)=0$ and
    \begin{align}
        \int_{-\infty}^{0} u^j K^*_{p,K}(u) \,du =\delta_{0,j}\,\,(0\leq j \leq p-1),\label{moments of equiv kernel of LPD}
    \end{align}
    where $\delta_{0,j}$ takes $1$ if $j=0$ and $0$ otherwise.
\end{lemma}

Lemma \ref{lemma: equiv kernel} shows that the equivalent kernel $K^*_{p,K}$ is the $p$-th order boundary kernel that satisfies the moment conditions of \citet[p.~244]{Gasser_etal:1985}, reflecting that the LPD estimator is boundary adaptive.
The property of $K^*_{p,K}(0)=0$ arises from our use of $\hat{F}$ as the ``dependent variable."
At the boundary, $\hat{F}(x_R) - \bm{r}_p(0)^\prime\hat{\bm{\beta}}$ shrinks very quickly \citep[Supplemental Appendix]{Cattaneo_etal:2020}, eliminating its contribution to the loss in \eqref{objective function}. 
As a result, the LPD asymptotically does not use information from the data at the boundary.

A well-known result from earlier studies on boundary kernel methods, such as \citet{Muller:1991}, suggests that there is no optimal weighting function among kernels satisfying \eqref{moments of equiv kernel of LPD} and $K^*_{p,K}(0)=0$.
However, we can show that an optimal weighting scheme exists within an extended class of weight functions that allows $K^*_{p,K}(0) \neq 0$.
The Online Appendix (Section S1.5) shows that it corresponds to the equivalent kernel of \citeauthor{Cheng_etal:1997}'s (\citeyear{Cheng_etal:1997}) local polynomial density estimator using the triangular kernel. Therefore, the closer $K_{p,k}^*$ is to this optimal weight, the more the MSE of the LPD estimator with kernel $K$ improves.

\begin{figure}
    \begin{center}
        \includegraphics[width=0.6\linewidth]{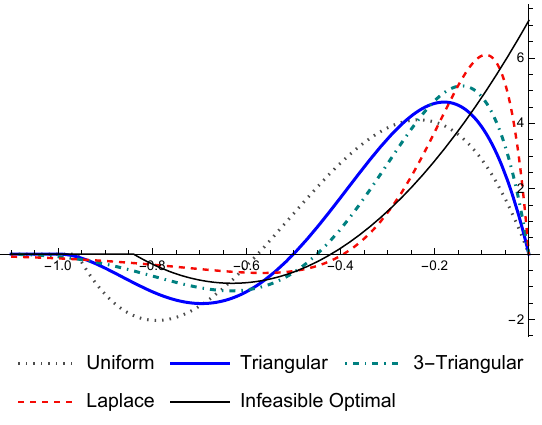}
    \caption{Equivalent Kernels}
     \label{fig: equivalent kernels}
    \end{center}

     \footnotesize
    \renewcommand{\baselineskip}{11pt}
    \textbf{Note:} $K^*_{2,K}$'s are scaled so that these kernels have the same asymptotic variance.
\end{figure}
Figure \ref{fig: equivalent kernels} displays several equivalent kernels from the generalized triangular family, scaled to have the same asymptotic variance ($\mathcal{V}_{2,K}$) for $p=2$, along with the infeasible optimal weight.
We find that the triangular and uniform kernels deviate substantially from the optimal weight.
This degree of deviation reflects the room for MSE-efficiency gains.

This deviation suggests that more centered kernels (i.e., those that place more weight near the evaluation point) can deliver improved MSE performance relative to commonly used kernels. 
Within the Beta and generalized triangular families, this can be achieved by choosing a larger $\gamma$ or $m$. 
Moreover, the efficiency gains do not appear to be exhausted within these families; consequently, among the kernels we consider, the Gaussian and Laplace kernels yield the best MSE performance.
As a result, Figure \ref{fig: equivalent kernels} shows that the Laplace kernel more closely matches the optimal weight and applies more weight near the evaluation point for a given variance, thereby delivering a lower MSE than standard kernels.\footnote{Figure \ref{fig: equivalent kernels} suggests that even relative to the Laplace kernel, there remains additional scope for efficiency improvements. Indeed, for example, we can construct a kernel function based on high-order polynomials, and a numerical search can find coefficients that achieve a smaller MSE than that of the Laplace kernel. Although appealing, such a weight may not be monotonically decreasing, and it can be computationally costly when performing estimation. Hence, our preferred choice is the Laplace for its simple form and satisfactory performance in numerical studies in Section \ref{sec: motivation}. A detailed study on this issue is postponed to future research.}

The analysis above offers a new perspective on standard KDE at boundary points. 
Classical contributions, such as \citet{Muller:1991} and \citet{Gasser_etal:1985}, develop boundary-kernel methods and find that an MSE-optimal kernel does not exist within the standard kernel class. 
However, they do not benchmark conventional kernels against the infeasible optimum within an enriched class of weighting functions. 
As a result, the fact that reasonable alternative kernels---for instance, the equivalent kernels of the Gaussian or Laplace kernels---can deliver numerically substantial MSE efficiency gains at the boundary has been largely overlooked in the kernel smoothing literature.
For example, \citeauthor{Muller:1991}'s (\citeyear{Muller:1991}) preferred kernel (\citealp[lines 13--31 on p.~524 and Section~3]{Muller:1991}) coincides with the equivalent kernel of the LPD estimator under the uniform kernel. 
Nevertheless, this kernel remains far from optimal: as shown in Table \ref{tab: kernel asy efficiency}, the equivalent kernels induced by the other kernels attain strictly better MSE efficiency. 
To the best of our knowledge, the magnitude and practical relevance of this efficiency gap have not been emphasized in the existing literature.

\begin{remark}[Intuitive Desirability]\label{rem: desirability}
    At first glance, employing the non-compactly supported kernels such as the Gaussian or Laplace seems counterintuitive, as they appear to ``use" the datum point \textit{far away from} the evaluation point.
    However, the analysis above shows that, after appropriate rescaling to keep the variance unchanged, these kernels asymptoticaly allow more weight to be placed \textit{near} the evaluation point than commonly used kernel functions (see Figure \ref{fig: equivalent kernels}), thereby improving the intuitive appeal of such a weighting scheme.
    Moreover, even in situations where one hesitates to trade off a larger bias for a smaller MSE, the asymptotic bias can be reduced at no cost by replacing the conventional kernel with the (scaled) Laplace kernel while keeping the bandwidth unchanged. (To see this, recall that the asymptotic variance of each scaled kernel is the same in Figure \ref{fig: equivalent kernels} under the fixed same bandwidth, while the asymptotic bias is smaller under the scaled Laplace kernel.)
\end{remark}

\begin{remark}[Interior Points]\label{rem:int}
    At this point, it is natural to ask whether a similar efficiency gain is also possible at interior points. 
    The answer is no. The mechanism behind the gains we document is inherently a boundary phenomenon and is closely related to boundary-kernel methods. 
    In particular, the condition $K_{p,K}^*(0)=0$ plays a central role: it arises because the discrepancy $\hat{F}(x_R)-\bm{r}_p(0)^\prime \hat{\bm{\beta}}$ shrinks quickly only at boundary points. 
    In fact, \citet{cattaneo2021local} shows that, at interior points, the uniform kernel is asymptotically MSE-optimal.
\end{remark}
\begin{remark}[Local Polynomial Regression]
    Similarly, one may wonder whether comparable efficiency gains can arise in the standard local polynomial regression setting. 
    Again, the answer is no, for essentially the same reason as in the previous remark. 
    In the regression setting, the outcome is the usual dependent variable rather than the empirical distribution function, and the resulting equivalent kernel places positive weight at the evaluation point regardless of whether the point is in the interior or at the boundary. 
    Consequently, the efficiency gains documented above do not carry over to standard local polynomial regression.
\end{remark}

\begin{remark}[A Word of Caution]
    As may be clear from the discussion in this section, our point is \textit{not} that unbounded-support kernels mechanically reduce variance. Indeed, this is the usual message in the standard kernel-smoothing literature. 
    Rather, we argue that, owing to a distinctive feature of the LPD estimator, Gaussian and Laplace kernels can outperform popular kernels in terms of MSE and testing efficiency even after accounting for the additional bias they may induce.
\end{remark}

\subsubsection{Inference}
We now turn to investigate why kernel choice plays a more prominent role in simple robust bias-corrected inference.

First, suppose the point estimation scenario with polynomial order $p=2$.
Under the MSE-optimal bandwidth $h_2^{\mathtt{MSE}}$, the bias and variance are balanced so that the variance is (asymptotically) proportional to $\mathcal{Q}_{2, K}$. 
Now, if we increase the polynomial order to $p=3$, then the variance increases by a factor of $\mathcal{V}_{3, K} / \mathcal{V}_{2, K}$. Hence, the constant factor of the variance under simple robust bias correction is given by $\mathcal{Q}_{2, K} \times \mathcal{V}_{3, K} / \mathcal{V}_{2, K}$, which a simple calculation shows to be equal to $\Theta_K$.
Therefore, in inference, the quantities $\mathcal{Q}_{2,K}$ and $\mathcal{V}_{3,K} / \mathcal{V}_{2,K}$ play central roles.

We have seen that $\mathcal{Q}_{2,K}$ is smaller for kernels with larger values of $m$ or $\gamma$.
Moreover, Figures \ref{fig: variability} and S3 (in the Online Appendix) suggest that larger $m$ or $\gamma$ also lead to a smaller increase in variability, $\mathcal{V}_{p+1,K}/\mathcal{V}_{p,K}$, mirroring the patterns observed in standard local polynomial smoothing (\citealp[pp.~77–79]{Fan_Gijbels:1996}).
\begin{figure}
    \begin{center}
        \includegraphics[width=0.6\linewidth]{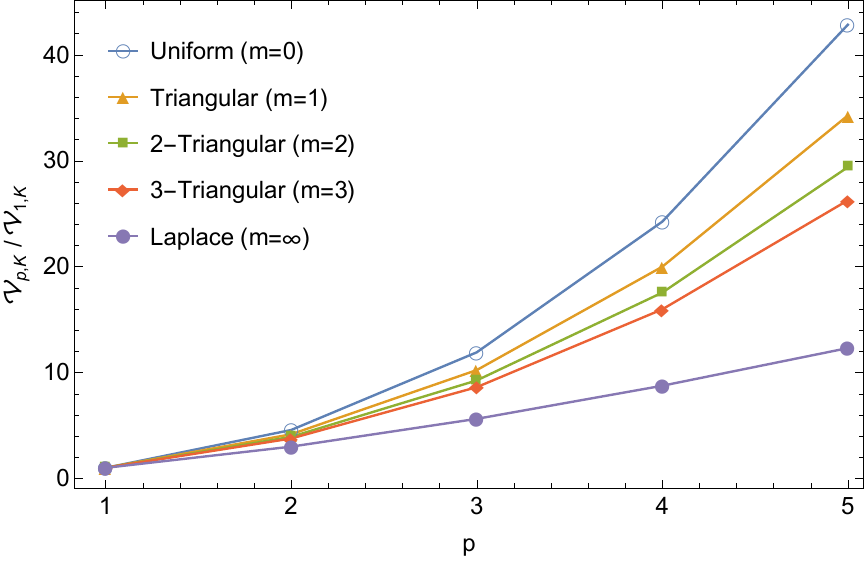}
    \caption{Increase of Variability}
     \label{fig: variability}
    \end{center}
\end{figure}
Hence, both factors (i.e., $\mathcal{Q}_{2, K}$ and $\mathcal{V}_{3, K} / \mathcal{V}_{2, K}$) become smaller when an appropriate kernel is chosen instead of the commonly used ones.
Consequently, the asymptotic variance under simple robust bias correction is much smaller with the Laplace kernel, illustrating a reason for the considerable efficiency gains reported in Table \ref{tab: kernel asy efficiency}.

\begin{remark}[Interior Points]
    Although modifying the kernel does not improve MSE efficiency at interior points (Remark \ref{rem:int}), the interior-point analogue of $\Theta_K$ can still vary meaningfully with kernel choice. 
    In the Online Appendix (Table S1), we find that the Laplace kernel performs best among the eight kernels considered, with the Gaussian kernel a close second.
    Hence, when statistical inference is the primary interest, the Gaussian or Laplace kernels would be recommended regardless of the evaluation points.
    
    We note that this observation does not contradict \citet{cattaneo2021local}, who develop an approach that attains the smallest asymptotic variance (i.e., the analogue of $\mathcal{Q}_{p,K}$) within the class of compactly supported kernels. 
    In contrast, our comparison accounts for the kernel-induced rescaling of the bandwidth by focusing on the simple robust bias-correction strategy (i.e., $\Theta_K$).
\end{remark}

\begin{remark}[McCrary Test]\label{rem: mccrary}
    In view of the shape of the equivalent kernels of the LPD estimator, $K_{p,K}^*(0)=0$ may not be very appealing, as it is counterintuitive, and it induces the pronounced oscillations characteristic of the equivalent kernels of the LPD estimator.
    To avoid this, one might be tempted to employ the alternative local linear density estimator of \citet[Section~2.2]{Cheng_etal:1997}, whose equivalent kernel satisfies $K_{p,K}^*(0)>0$, thereby yielding a smoother shape than that of the LPD estimator. 
    Asymptotically, or from a theoretical standpoint, this estimator may perform well.
    However, a drawback of \citeauthor{Cheng_etal:1997}'s estimator is that it requires pre-binning, which reduces the effective sample size in finite samples. 
    Indeed, \citet[p.~1454]{Cattaneo_etal:2020} report that the \citeauthor{McCrary:2008} test, which is based on the \citeauthor{Cheng_etal:1997} estimator, may suffer from low power. 
    Moreover, \citet{Kuehnle_etal:2021} find that test results can be sensitive to the pre-binning specification. 
    Accordingly, we recommend using the LPD framework while tailoring the kernel choice to the research objective. 
    For instance, when statistical power is a primary concern, LPD with the Laplace or Gaussian kernel is a particularly attractive option.
\end{remark}

\section{Finite Sample Variance}\label{sec: finite}
Thus far, we have examined how the efficiency properties of the LPD estimator depend on the choice of kernel function in the asymptotic sense.
However, as observed in Section \ref{subsec: simu small}, there remains another important issue: the explosion of finite-sample variance.
This section provides a theoretical explanation for this numerical finding.
In particular, we formally show that the finite-sample variance of the LPD estimator is unbounded when compactly supported kernels are used, whereas it remains bounded when kernels with unbounded support are employed. 

At first glance, the small sample property may seem less important, as the sample size is often large in recent empirical studies. 
However, the finite sample variance property is especially relevant in one of the most important applications: manipulation testing in RD designs.
To motivate readers, we begin by explaining the unique feature of the score manipulation problem.

\subsection{A Leading Example: Score Manipulation in RD Designs}\label{sec finite manip}
In the RD analysis, score manipulation is a central concern, and most of the recent RD studies report results from the LPD-based density discontinuity test to assess the presence of manipulation.
In many empirical applications, manipulation is expected to be one-sided---i.e., from below the cutoff to above it (see \citealp{Gerard_etal:2020}).
For a typical educational example, if a generous professor awards a few extra points to students whose test scores are 58 or 59, where 60 is the minimum passing threshold, then one-sided score manipulation arises.

An important implication of such one-sided manipulation is that the sample size below the cutoff necessarily decreases, as illustrated in Figure \ref{fig:manip}. 
This reduction in effective sample size makes the finite-sample variance behavior observed in Section \ref{subsec: simu small} particularly relevant: The one-sided manipulation always inflates the variance of the LPD estimator on one side. 

When the extent of manipulation is large, variance reduction becomes particularly important.
A greater manipulation further reduces the sample size below the cutoff. Consequently, test performance does not necessarily improve with a larger effect size (or a larger jump); rather, estimation variance can become so severe that the test effectively breaks down.

This contrasts sharply with standard inference settings, in which a larger effect size typically increases detection probability, making variance reduction only of secondary importance.
This is not necessarily the case in the context of manipulation testing, making it crucial to control the finite-sample variance.
\begin{figure}[t]
    \begin{center}
        \includegraphics{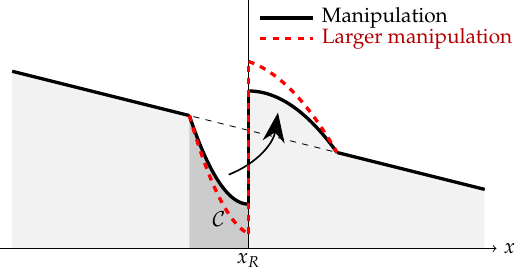}
        \caption{Manipulation and Sample Size around the Cutoff}
        \label{fig:manip}
    \end{center}

    \footnotesize
    \renewcommand{\baselineskip}{11pt}
    \textbf{Note:} The thick lines are the density of the running variable, and the thin dashed line represents the density before manipulation. The manipulation happens in the region $\mathcal{C}$.
\end{figure}

\subsection{Finite-Sample Variance and Kernel Support}\label{subsec: finite sample variance}
\subsubsection{Variance Inflation under Compactly Supported Kernels}
We begin by providing a theoretical explanation for the variance inflation under the popular kernels, including the uniform and triangular, as observed in Figure \ref{fig: finite variance}.
We make the following assumptions.
\begin{assF}\label{assumption kernel compact}
    The kernel function is non-negative, bounded, and compactly supported.
\end{assF}
\begin{assF}\label{assumption for variance explosion}
    $f(x)$ is supported by a bounded region, $\mathcal{S}=[x_L,x_R)$, and $0<\delta\leq f(x) \leq \Delta <\infty$ on $\mathcal{S}$.
\end{assF}
\begin{assF}\label{assumption no tie}
    There exists no tie in $X_{1:n}$.
\end{assF}
Assumptions \ref{assumption for variance explosion} and \ref{assumption no tie} are standard regularity conditions.
Note that Assumption \ref{assumption no tie} holds with probability one when $f$ is continuous. 
Under these assumptions, we have the next theorem: 
\begin{theorem}\label{theorem: infinite variance}
    Under Assumptions \ref{assumption kernel compact}, \ref{assumption for variance explosion}, and \ref{assumption no tie}, the LPD estimator $\hat{f}(x_R)$ of degree $p=1$ does not have a second moment.
\end{theorem}
\begin{remark}
The statement should also hold for $p\geq2$, as the variance will inflate as the number of parameters to be estimated increases. 
Proving this formally seems to require an algebraic effort, but we provide proof for the case when $p=2$ in the Online Appendix to illustrate that the essential problem is the same.
\end{remark}

Theorem \ref{theorem: infinite variance} states that we cannot have a bounded variance of the LPD estimator.
This provides a theoretical explanation for the variance inflation under the compactly supported triangular kernel, observed in Figure \ref{fig: finite variance}.

A heuristic explanation of why the variance inflates is as follows.
Assume that the kernel function is compactly supported and fix $h$ and $p$.
In finite samples, the local sample size within the bandwidth $h$ is not infinitely many, rather it can be only a few compared to the polynomial degree $p$.\footnote{In practice, the popular bandwidth selectors in recent empirical economics are often of the plug-in type. When such a selector is used, the bandwidth takes the form of $h=C\times n^{-1/5}$, where $n$ is the \textit{whole} sample size. Consequently, the bandwidth can be small in terms of the \textit{local} sample size, especially when the whole sample size $n$ is large---as is typically the case in recent empirical applications.}
In such a case, the estimates can be unstable in general.
Here, importantly and in contrast to the standard KDE, local polynomial fitting allows the slope term of $\hat{\bm{\beta}}$ (i.e., $\hat{f}(x_R)$) to take arbitrarily large values, especially when observations are located very closely.
As an extreme case, consider $p=1$ and only two data points within the bandwidth. The LPD estimate then corresponds to the slope of the line connecting these two points (i.e., $(\hat{F}(X_{2}) -\hat{F}(X_{1}))/((X_{2} - X_{1})$ when $X_2>X_1$), which diverges as the points become closer.
As a result, $\hat{f}(x_R)$ can become excessively large with a certain probability, and the finite-sample variance explodes.

Mathematically (but still heuristically), this intuition can be formulated as follows.
Using the law of total variance, we have that 
\begin{align*}
    \V{\hat{f}(x_R)} &\geq \E{\V{\hat{f}(x_R) \mid \mathbf{1}\left\{n_0 \leq p+2\right\}}}\geq \V{\hat{f}(x_R)  \mid  n_0 \leq p+2} 
    \underbrace{\P{n_0 \leq p+2}}_{>0},
\end{align*}
where $n_0$ is the local sample size within the bandwidth, that is, the number of observations in $[x_R - h,x_R)$.
Here, the event $\{n_0 \leq p+2\}$ can be viewed as a mathematical formulation of the intuition that ``the effective sample size is small compared to the model flexibility."
Recall that the LPD estimator is a coefficient of a polynomial regression model on $[x_R - h,x_R)$.
Then, we can deduce from the classic linear regression theory (e.g., \citealp{Kinal:1980}) that the conditional variance on the right-hand side diverges, thereby leading to $\mathbb{V}[\hat{f}(x_R)] = \infty$.

The preceding heuristics do not rely on any LPD-specific features, but rather on the general nature of local polynomial fitting.
In this sense, this erratic behavior is not unique to the LPD estimator but inherent to local polynomial fitting in general. 
Indeed, \citet{Seifert_Gasser:1996} shows that the local linear regression estimator with a compactly supported kernel has infinite finite-sample variance under homoskedasticity. 
Hence, our result shows that the LPD estimator also inherits an \textit{undesirable} feature associated with local polynomial smoothing techniques. (See the Online Appendix for further details.)

\subsubsection{Variance Suppression under Unbounded Supported Kernels}
A simple remedy for this result is available: use a non-compactly supported kernel, as suggested by Figure \ref{fig: finite variance}.\footnote{Another possible solution is to manually enlarge the bandwidth. However, doing so may introduce room for specification search and somewhat undermine one of the appealing features of the LPD approach, or its fully automatic implementation. Moreover, when conducting inference, a manually increased bandwidth can be too large, potentially rendering the procedure invalid (see related discussion in \citealp{Calonico_etal:2014}).}
Intuitively, the infinite variance problem does not arise if we use a non-compactly supported kernel since the ``local" sample size $n_0$ is equal to the whole sample size $n$.
Put differently, when a kernel function is non-compactly supported, the LPD estimator can be seen as the usual weighted least squares that uses the whole sample, and we can expect it to exhibit much stable behavior. 
Formally, we have the following result.

\begin{assF}\label{assumption kernel non-compact}
    The kernel function is bounded, positive everywhere, centered at zero, and monotonically decreasing from the peak.
\end{assF}
\begin{theorem}\label{theorem: finite variance}
    Assume $n\geq6$. Under Assumptions \ref{assumption for variance explosion}, \ref{assumption no tie}, and \ref{assumption kernel non-compact}, the LPD estimator $\hat{f}(x_R)$ of degree $p=1$ has a second moment.
\end{theorem}
This result also mirrors the finding of \citet[Theorem 2]{Seifert_Gasser:1996}.

Theorems \ref{theorem: infinite variance} and \ref{theorem: finite variance} suggest that the use of unbounded-support kernels, such as the Laplace or Gaussian kernel, is advisable. Note that, in contrast to local polynomial regression, these kernels remain a reasonable choice even asymptotically, as shown in Section \ref{sec: efficiency}.
This variance-stabilizing property of unbounded kernels is empirically relevant, particularly in discontinuity detection settings, as demonstrated in Section \ref{sec finite manip}.

\section{Practical Recommendations}\label{sec: recommendation}
Our message for empirical researchers is simple: \textit{kernel choice is an important design decision} for estimation and inference based on the popular LPD estimator, even though this conclusion contrasts with the conventional view in the kernel-smoothing literature that kernel selection is largely irrelevant to numerical performance.

Two sets of results motivate practical guidance. First, our finite-sample theory indicates that compactly supported kernels can be vulnerable to variance inflation, whereas unbounded-support kernels such as the Gaussian and Laplace are naturally guarded against this problem. Second, our large-sample analysis shows that, at boundary points, the Laplace kernel is particularly attractive in terms of MSE, CI length, and statistical power.

That said, the choice of kernel ultimately remains a matter of researcher preference and empirical priorities. Unlike bandwidth selection, for which the optimal bandwidth is defined and a well-developed automatic rule exists \citep{Cattaneo_etal:2020}, the infeasible optimal weighting scheme is not achieved by any kernel function. Accordingly, we recommend selecting the kernel to align with the main empirical goal, using the following rules of thumb.

\textbf{Boundary inference.}
When accurate boundary estimation is the primary objective---and especially when inference at the boundary is central---the Laplace kernel is a strong default choice. In our analysis, it delivers a favorable MSE and produces shorter intervals and higher power than commonly used compactly supported kernels.

\textbf{Applications requiring good performance both in the interior and at the boundary.}
In some empirical settings, researchers may be interested in both boundary behavior and interior fit. In such cases, the Laplace kernel may be less appealing because, at interior points, it can be substantially less MSE-efficient than compactly supported kernels such as the uniform or triangular kernels (see Table S1 in the Online Appendix). 
A pragmatic compromise may be the Gaussian kernel: it tends to improve boundary behavior relative to standard compact kernels while remaining reasonably competitive in the interior.
Furthermore, the Gaussian kernel can avoid the finite-sample variance explosion.

\textbf{Bias-sensitive applications.}
Finally, if bias is the primary concern, the Laplace kernel may be less attractive in some applications, as its MSE-efficiency gains can come at the cost of increased bias. In such cases, the Gaussian kernel may be a good alternative: it tends to have milder bias while still benefiting from unbounded support and is widely used in empirical work in standard settings.

\section{Concluding Remarks}\label{sec: conclusion}
Motivated by their usefulness and intuitive simplicity, recent econometrics and statistics literature extends local polynomial methods beyond the standard LPD setting to the estimation and inference of densities and related objects \citep{cattaneo2022conditional, cattaneo2021local}. These extensions have also attracted increasing attention in empirical work. We conjecture that similar considerations regarding kernel choice arise in these settings as well.

\bibliographystyle{apalike} 
\bibliography{refs}

\begin{thebibliography}{}

\bibitem[Bertrand et~al., 2015]{Bertrand_etal:2015}
Bertrand, M., Kamenica, E., and Pan, J. (2015).
\newblock { Gender Identity and Relative Income within Households }.
\newblock {\em The Quarterly Journal of Economics}, 130(2):571--614.

\bibitem[Breunig et~al., 2024]{Breunig_etal:2024}
Breunig, R., Deutscher, N., and Hamilton, S. (2024).
\newblock {Rounded Up: Using round numbers to identify tax evasion}.
\newblock {\em Journal of Public Economics}, 238:105195.

\bibitem[Britto et~al., 2022]{Britto_etal:ECTA}
Britto, D. G.~C., Pinotti, P., and Sampaio, B. (2022).
\newblock {The Effect of Job Loss and Unemployment Insurance on Crime in Brazil}.
\newblock {\em Econometrica}, 90(4):1393--1423.

\bibitem[Calonico et~al., 2018]{calonico2018effect}
Calonico, S., Cattaneo, M.~D., and Farrell, M.~H. (2018).
\newblock {On the Effect of Bias Estimation on Coverage Accuracy in Nonparametric Inference}.
\newblock {\em Journal of the American Statistical Association}, 113(522):767--779.

\bibitem[Calonico et~al., 2022]{calonico2022coverage}
Calonico, S., Cattaneo, M.~D., and Farrell, M.~H. (2022).
\newblock {Coverage Error Optimal Confidence Intervals for Local Polynomial Regression}.
\newblock {\em Bernoulli}, 28(4):2998--3022.

\bibitem[Calonico et~al., 2014]{Calonico_etal:2014}
Calonico, S., Cattaneo, M.~D., and Titiunik, R. (2014).
\newblock {Robust Nonparametric Confidence Intervals for Regression-Discontinuity Designs}.
\newblock {\em Econometrica}, 82(6):2295--2326.

\bibitem[Cattaneo et~al., 2024a]{cattaneo2022conditional}
Cattaneo, M.~D., Chandak, R., Jansson, M., and Ma, X. (2024a).
\newblock {Boundary Adaptive Local Polynomial Conditional Density Estimators}.
\newblock {\em Bernoulli}, 30(4):3193--3223.

\bibitem[Cattaneo et~al., 2015]{Cattaneo_etal:2015}
Cattaneo, M.~D., Frandsen, B.~R., and Titiunik, R. (2015).
\newblock {Randomization Inference in the Regression Discontinuity Design: An Application to Party Advantages in the U.S. Senate}.
\newblock {\em Journal of Causal Inference}, 3(1):1--24.

\bibitem[Cattaneo et~al., 2024b]{cattaneo2023practical}
Cattaneo, M.~D., Idrobo, N., and Titiunik, R. (2024b).
\newblock {\em {A Practical Introduction to Regression Discontinuity Designs: Extensions}}.
\newblock Elements in Quantitative and Computational Methods for the Social Sciences. Cambridge University Press.

\bibitem[Cattaneo et~al., 2018]{Cattaneo_etal:2018}
Cattaneo, M.~D., Jansson, M., and Ma, X. (2018).
\newblock {Manipulation Testing Based on Density Discontinuity}.
\newblock {\em The Stata Journal}, 18(1):234--261.

\bibitem[Cattaneo et~al., 2020]{Cattaneo_etal:2020}
Cattaneo, M.~D., Jansson, M., and Ma, X. (2020).
\newblock {Simple Local Polynomial Density Estimators}.
\newblock {\em Journal of the American Statistical Association}, 115(531):1449--1455.

\bibitem[Cattaneo et~al., 2022]{Cattaneo_etal:2022}
Cattaneo, M.~D., Jansson, M., and Ma, X. (2022).
\newblock {lpdensity: Local Polynomial Density Estimation and Inference}.
\newblock {\em Journal of Statistical Software}, 101(1):1--25.

\bibitem[Cattaneo et~al., 2024c]{cattaneo2021local}
Cattaneo, M.~D., Jansson, M., and Ma, X. (2024c).
\newblock {Local Regression Distribution Estimators}.
\newblock {\em Journal of Econometrics}, 240(2):105074.

\bibitem[Cattaneo et~al., 2023]{Cattaneo_etal:2023med}
Cattaneo, M.~D., Keele, L., and Titiunik, R. (2023).
\newblock {A Guide to Regression Discontinuity Designs in Medical Applications}.
\newblock {\em Statistics in Medicine}, 42(24):4484--4513.

\bibitem[Cattaneo and Titiunik, 2022]{Cattaneo_Titiunik:2022}
Cattaneo, M.~D. and Titiunik, R. (2022).
\newblock {Regression Discontinuity Designs}.
\newblock {\em Annual Review of Economics}, 14(1):821--851.

\bibitem[Chen et~al., 2023]{Chen_etal:JDevEcon}
Chen, W.-L., Lin, M.-J., and Yang, T.-T. (2023).
\newblock {Curriculum and National Identity: Evidence from the 1997 Curriculum Reform in Taiwan}.
\newblock {\em Journal of Development Economics}, 163:103078.

\bibitem[Cheng et~al., 1997]{Cheng_etal:1997}
Cheng, M.-Y., Fan, J., and Marron, J.~S. (1997).
\newblock {On Automatic Boundary Corrections}.
\newblock {\em The Annals of Statistics}, 25(4):1691 -- 1708.

\bibitem[Cline, 1988]{Cline:1988}
Cline, D. B.~H. (1988).
\newblock {Admissible Kernel Estimators of a Multivariate Density}.
\newblock {\em The Annals of Statistics}, 16(4):1421--1427.

\bibitem[Collin and Talbot, 2023]{Collin_Talbot:2023}
Collin, M. and Talbot, T. (2023).
\newblock {Are Age-of-Marriage Laws Enforced? Evidence from Developing Countries}.
\newblock {\em Journal of Development Economics}, 160:102950.

\bibitem[Connolly and Haeck, 2022]{Connolly_Haeck:JLaborEcon}
Connolly, M. and Haeck, C. (2022).
\newblock {Nonlinear Class Size Effects on Cognitive and Noncognitive Development of Young Children}.
\newblock {\em Journal of Labor Economics}, 40(S1):S341--S382.

\bibitem[Dasgupta et~al., 2022]{Dasgupta_etal:REStat}
Dasgupta, U., Mani, S., Sharma, S., and Singhal, S. (2022).
\newblock {Effects of Peers and Rank on Cognition, Preferences, and Personality}.
\newblock {\em The Review of Economics and Statistics}, 104(3):587--601.

\bibitem[David and Nagaraja, 2004]{david2004order}
David, H.~A. and Nagaraja, H.~N. (2004).
\newblock {\em {Order Statistics}}.
\newblock John Wiley \& Sons.

\bibitem[{De Benedetto} et~al., 2025]{DeBenedetto_etal:2025}
{De Benedetto}, M.~A., {De Paola}, M., Scoppa, V., and Smirnova, J. (2025).
\newblock Erasmus program and labor market outcomes: Evidence from a fuzzy regression discontinuity design.
\newblock {\em Labour Economics}, 93:102675.

\bibitem[Eggers et~al., 2021]{Eggers_etal:2021JME}
Eggers, A.~C., Ellison, M., and Lee, S.~S. (2021).
\newblock {The Economic Impact of Recession Announcements}.
\newblock {\em Journal of Monetary Economics}, 120:40--52.

\bibitem[Elliott et~al., 2022]{Elliott_etal:2022}
Elliott, G., Kudrin, N., and Wüthrich, K. (2022).
\newblock {Detecting p-Hacking}.
\newblock {\em Econometrica}, 90(2):887--906.

\bibitem[Epanechnikov, 1969]{Epanechnikov:1969}
Epanechnikov, V.~A. (1969).
\newblock {Non-Parametric Estimation of a Multivariate Probability Density}.
\newblock {\em Theory of Probability \& Its Applications}, 14(1):153--158.

\bibitem[Fan and Gijbels, 1995]{Fan_Gijbels:1995}
Fan, J. and Gijbels, I. (1995).
\newblock {Adaptive Order Polynomial Fitting: Bandwidth Robustification and Bias Reduction}.
\newblock {\em Journal of Computational and Graphical Statistics}, 4(3):213--227.

\bibitem[Fan and Gijbels, 1996]{Fan_Gijbels:1996}
Fan, J. and Gijbels, I. (1996).
\newblock {\em {Local Polynomial Modelling and Its Applications}}.
\newblock Chapman \& Hall/CRC.

\bibitem[Forderer and Burtch, 2025]{Forderer_Burtch:2025}
Forderer, J. and Burtch, G. (2025).
\newblock {Estimating Career Benefits from Online Community Leadership: Evidence from Stack Exchange Moderators}.
\newblock {\em Management Science}, 71(3):2443--2466.

\bibitem[Gasser et~al., 1985]{Gasser_etal:1985}
Gasser, T., M\"uller, H.-G., and Mammitzsch, V. (1985).
\newblock {Kernels for Nonparametric Curve Estimation}.
\newblock {\em Journal of the Royal Statistical Society: Series B (Methodological)}, 47(2):238--252.

\bibitem[Gerard et~al., 2020]{Gerard_etal:2020}
Gerard, F., Rokkanen, M., and Rothe, C. (2020).
\newblock {Bounds on Treatment Effects in Regression Discontinuity Designs with a Manipulated Running Variable}.
\newblock {\em Quantitative Economics}, 11(3):839--870.

\bibitem[Gorrín et~al., 2023]{Gorrin_etal:2023JIE}
Gorrín, J., Morales-Arilla, J., and Ricca, B. (2023).
\newblock {Export Side Effects of Wars on Organized Crime: The Case of Mexico}.
\newblock {\em Journal of International Economics}, 144:103775.

\bibitem[Granovsky and M\"uller, 1991]{Granovsky_Muller:1991}
Granovsky, B.~L. and M\"uller, H.~G. (1991).
\newblock {Optimizing Kernel Methods: A Unifying Variational Principle}.
\newblock {\em International Statistical Review / Revue Internationale de Statistique}, 59(3):373--388.

\bibitem[Hansen, 2022]{Hansen:2022_prob}
Hansen, B.~E. (2022).
\newblock {\em {Probability \& Statistics for Economists}}.
\newblock Princeton University Press.

\bibitem[He et~al., 2020]{He_etal:2020}
He, G., Wang, S., and Zhang, B. (2020).
\newblock {Watering Down Environmental Regulation in China}.
\newblock {\em The Quarterly Journal of Economics}, 135(4):2135--2185.

\bibitem[Keefer and Vlaicu, 2025]{Keefer_Vlaicu:2025}
Keefer, P. and Vlaicu, R. (2025).
\newblock Voting age, information experiments, and political engagement: Evidence from a general election.
\newblock {\em Journal of Development Economics}, 174:103458.

\bibitem[Khanna, 2023]{Khanna:JPE}
Khanna, G. (2023).
\newblock {Large-Scale Education Reform in General Equilibrium: Regression Discontinuity Evidence from India}.
\newblock {\em Journal of Political Economy}, 131(2):549--591.

\bibitem[Kinal, 1980]{Kinal:1980}
Kinal, T.~W. (1980).
\newblock {The Existence of Moments of k-Class Estimators}.
\newblock {\em Econometrica}, 48(1):241--249.

\bibitem[Kuehnle et~al., 2021]{Kuehnle_etal:2021}
Kuehnle, D., Oberfichtner, M., and Ostermann, K. (2021).
\newblock {Revisiting Gender Identity and Relative Income within Households: A Cautionary Tale on the Potential Pitfalls of Density Estimators}.
\newblock {\em Journal of Applied Econometrics}, 36(7):1065--1073.

\bibitem[Lindo et~al., 2010]{Lindo_etal:2010}
Lindo, J.~M., Sanders, N.~J., and Oreopoulos, P. (2010).
\newblock {Ability, Gender, and Performance Standards: Evidence from Academic Probation}.
\newblock {\em American Economic Journal: Applied Economics}, 2(2):95--117.

\bibitem[McCrary, 2008]{McCrary:2008}
McCrary, J. (2008).
\newblock {Manipulation of the Running Variable in the Regression Discontinuity Design: A Density Test}.
\newblock {\em Journal of Econometrics}, 142(2):698--714.

\bibitem[M\"uller, 1984]{Muller:1984}
M\"uller, H.-G. (1984).
\newblock {Smooth Optimum Kernel Estimators of Densities, Regression Curves and Modes}.
\newblock {\em The Annals of Statistics}, 12(2):766 -- 774.

\bibitem[M\"uller, 1991]{Muller:1991}
M\"uller, H.-G. (1991).
\newblock {Smooth Optimum Kernel Estimators Near Endpoints}.
\newblock {\em Biometrika}, 78(3):521--530.

\bibitem[Pinkus, 2009]{Pinkus:2009}
Pinkus, A. (2009).
\newblock {\em {Totally Positive Matrices}}.
\newblock Cambridge Tracts in Mathematics. Cambridge University Press.

\bibitem[Seifert and Gasser, 1996]{Seifert_Gasser:1996}
Seifert, B. and Gasser, T. (1996).
\newblock {Finite-Sample Variance of Local Polynomials: Analysis and Solutions}.
\newblock {\em Journal of the American Statistical Association}, 91(433):267--275.

\bibitem[van~der Vaart, 1998]{vanderVaart:1998}
van~der Vaart, A.~W. (1998).
\newblock {\em {Asymptotic Statistics}}.
\newblock Cambridge Series in Statistical and Probabilistic Mathematics. Cambridge University Press.

\bibitem[van Eeden, 1985]{vanEeden:1985}
van Eeden, C. (1985).
\newblock {Mean Integrated Squared Error of Kernel Estimators When the Density and Its Derivative are not Necessarily Continuous}.
\newblock {\em Annals of the Institute of Statistical Mathematics}, 37:461--472.

\end{thebibliography}

\newpage
\setcounter{section}{0}
\setcounter{page}{1}
\renewcommand{\thepage}{S\arabic{page}}
\setcounter{equation}{0}
\renewcommand{\theequation}{S.\arabic{equation}}
\setcounter{lemmax}{0}
\renewcommand{\thelemmax}{S\arabic{lemmax}}
\setcounter{theoremx}{0}
\renewcommand{\thetheoremx}{S\arabic{theoremx}}
\setcounter{propx}{0}
\renewcommand{\thepropx}{S\arabic{propx}}
\setcounter{corx}{0}
\renewcommand{\thetheoremx}{S\arabic{corx}}
\setcounter{table}{0}
\renewcommand{\thesection}{S\arabic{section}}
\setcounter{table}{0}
\renewcommand{\thetable}{S\arabic{table}}
\setcounter{figure}{0}
\renewcommand{\thefigure}{S\arabic{figure}}
\setcounter{remarkx}{0}
\renewcommand{\theremarkx}{S\arabic{remarkx}}

\renewcommand{\theassumption}{S.\arabic{assumption}}


\makeatletter
\renewcommand*{\@fnsymbol}[1]{\ensuremath{\ifcase#1\or \flat\or * \else\@ctrerr\fi}}
\makeatother

\begin{center}
{\Large\bf Online Appendix for ``Kernel Choice Matters for Local Polynomial Density Estimators at Boundaries"}
\end{center}
\begin{center}
    {\large Shunsuke Imai \& Yuta Okamoto}
\end{center}

\begin{abstract}
    Section \ref{sec: oa proof} provides the proofs of our main results.
    Section \ref{sec: oa equiv} outlines the derivation of the equivalent kernel.
    Section \ref{sec: oa int} discusses the interior points.
    Section \ref{sec: oa fig} collects omitted figures.
\end{abstract}

\section{Proofs}\label{sec: oa proof}
\subsection{Proofs of Theorems}\label{appendix proof}
\begin{proof}[Proof of Theorem \ref{theorem: infinite variance}]
    Put $\mathbb{E}_{2}[\cdot]=\mathbb{E}[\cdot \mid n_0=2]$. 
    Let $X_{(1)},X_{(2)},\dots, X_{(n-1)}, X_{(n)}$ be ordered observations. 
    We have
    \begin{align}
        \mathbb{E}_{2}\left[\hat{f}(x_R)^2\right]
        &=
        \mathbb{E}_{2}\left[\left(\frac{\hat{F}(X_{(n)}) -\hat{F}(X_{(n-1)})}{(X_{(n)} - x_R) - (X_{(n-1)} - x_R)}\right)^2\right],\label{fit line LPD}
    \end{align}
    where we utilize the fact that the coefficient can be considered as the ``slope" of the straight line through two data points. Note that this holds regardless of the kernel function as long as Assumption \ref{assumption kernel compact} is satisfied.
    Then, noticing that $\hat{F}(X_{(n)}) = 1$ and $\hat{F}(X_{(n-1)})= 1 - 1/n$, we have
    \begin{align}
        \mathbb{E}_{2}\left[\hat{f}(x_R)^2\right] &=
        \mathbb{E}_{2}\left[\frac{\left\{1 - \left(1-1/n\right)\right\}^2}{\left\{X_{(n)} - X_{(n-1)}\right\}^2}\right]
        \geq
        \mathbb{E}_{2}\left[\frac{\left(1/n\right)^2}{\left\{x_R - X_{(n-1)}\right\}^2 }\right].\label{similar form LPD}
    \end{align}
    Now, it follows from the definition of the conditional expectation that
    \begin{align}
        \mathbb{E}_{2}&\left[\frac{\left(1/n\right)^2}{\left\{x_R - X_{(n-1)}\right\}^2 }\right]=
        \frac{1}{n^2}\times\frac{1}{P_1}
        \int_{\mathcal{S}} \frac{1}{(x_R- z)^2}\tilde{\mathbb{P}}(z)\,dz, \label{eq conditional ex2}
    \end{align}
    where $P_1=\mathbb{P}[n_0=2]$ and $ \tilde{\mathbb{P}}(z)$ is defined by
    \begin{align*}
        \tilde{\mathbb{P}}(z)\, dz& =
        \mathbb{P}\left[n_0=2\,\mathrm{and}\, X_{(n-1)} \in [z-dz/2, z+dz/2]\right].
    \end{align*}
    The joint density of all $n$ ordered statistics is given by $n! f(x_1)f(x_2)\cdots f(x_n)$ with $x_1 < x_2 < \dots < x_n$ \cite[p.12]{david2004order}. Hence, we have, with $\tilde{f}(z)=f_X(z)1\{x_R-h \le z\}$, that
    \begin{align*}
        \tilde{\mathbb{P}}(z)\, dz&= 
        dz\int_{z}^{x_R} \int_{x_L}^{x_R-h}\int_{x_L}^{x_{n-2}}\dots \int_{x_L}^{x_3} \int_{x_L}^{x_2} n!f(x_1)\cdots f(x_{n-2})\tilde{f}(z)f(x_n) \, dx_1\dots dx_{n-2}dx_n\\
        & \geq dz\cdot n!\delta^{n-1}\tilde{f}(z)\int_{z}^{x_R} \int_{x_L}^{x_R-h}\int_{x_L}^{x_{n-2}}\dots \int_{x_L}^{x_3} \int_{x_L}^{x_2} 1 \,dx_1\dots dx_{n-2}dx_n\\
        & = dz\cdot n!\delta^{n-1}\tilde{f}(z) C(n,h,x_L,x_R) \int_{z}^{x_R} 1 \,dx_n \\
        & = dz\cdot n!\delta^{n-1}\tilde{f}(z) C(n,h,x_L,x_R)\cdot (x_R-z),
    \end{align*}
    where the inequality follows from Assumption \ref{assumption for variance explosion}, and $C(n,h,x_L,x_R)\in(0,\infty)$ is some constant depending on $n$, $h$, $x_L$, and $x_R$.
    This implies that 
    \begin{align*}
        \frac{1}{P_1 n^2} \int_{\mathcal{S}} \frac{1}{(x_R- z)^2} \tilde{\mathbb{P}}(z)\, dz
        &\ge \frac{n!}{P_1n^2} \delta^{n-1}C(n,h,x_L,x_R) \int_{x_L}^{x_R} \frac{1}{(x_R-z)} \tilde{f}_X(z)\, dz \\
        & = \frac{n!}{P_1n^2} \delta^{n-1}C(n,h,x_L,x_R)  \int_{x_R-h}^{x_R} \frac{1}{(x_R-z)} f_X(z)\, dz\\
        & \ge \frac{n!}{P_1n^2} \delta^{n}C(n,h,x_L,x_R)  \int_{x_R-h}^{x_R} \frac{1}{(x_R-z)}\, dz,
    \end{align*}
    where we let $\int_{x_L}^{x_R}$ denotes $\lim_{x\rightarrow x_R} \int_{x_L}^{x}$.
    Then, combined with \eqref{similar form LPD} and \eqref{eq conditional ex2}, we obtain that
    \begin{align*}
        \mathbb{E}\left[\hat{f}(x_R)^{2}\right]
        \geq
        \mathbb{E}_{2}\left[\hat{f}(x_R)^{2}\right]
        \mathbb{P}\left[n_0=2\,\right]\geq
        \frac{n! \delta^n C(n,h,x_L,x_R)}{n^2} \int_{x_R-h}^{x_R}\frac{1}{x_R - z} \,dz
        =\infty,
    \end{align*}
    which completes the proof.
\end{proof}

\begin{proof}[Proof of Theorem \ref{theorem: finite variance}]
For notational simplicity, we write $u_{i,h} \coloneqq ({X_i-x})/{h}$, $K_{i,h} \coloneqq K( u_{i,h} )$, $\tilde{\bm{S}}_{(p,h,x)} \coloneqq ({nh})^{-1} \sum_{i=1}^n \bm{r}_p(u_{i,h})\bm{r}_p(u_{i,h})^\prime K_{i,h}$, $\bm{\tilde{\Omega}}_{(p,h,K)} \coloneqq {h}^{-1} [K_{1,h}\bm{r}_p(u_{1,h}), \ldots,  K_{n,h}\bm{r}_p(u_{n,h})]$, and $\bm{\hat{F}} \coloneqq (\hat{F}(X_1), \dots, \hat{F}(X_n))^\prime$.
With no loss of generality, we can set $x_R=0$ and $x_L= -L$ with $L>0$. 
The LPD estimator with polynomial order $p=1$ at the boundary point $x_R=0$ is given by $\hat{f}(0)  = ({nh})^{-1}\bm{e}_1^\prime \tilde{\bm{S}}_{(1,h,0)}^{-1} \bm{\tilde{\Omega}}_{(1,h,0)} \bm{\hat{F}}$.
To prove the theorem, we shall bound the second moment of this estimator,
\begin{align*}
        \mathbb{E}\left[\hat{f}(0)^2\right] = 
        \mathbb{E}\left[  \left( \frac{1}{nh}\bm{e}_1^\prime \tilde{\bm{S}}_{(1,h,0)}^{-1} \bm{\tilde{\Omega}}_{(1,h,0)} \bm{\hat{F}} \right)^2 \right].
\end{align*}
First, we rewrite $\tilde{\bm{S}}_{(1,h,0)}^{-1}$ as $\tilde{\bm{S}}_{(1,h,0)}^{-1} = |\tilde{\bm{S}}_{(1,h,0)}|^{-1} \bm{\tilde{M}}_{(1,h,0)}$ with
\begin{align*}
    |\tilde{\bm{S}}_{(1,h,0)}| \coloneqq \det(\tilde{\bm{S}}_{(1,h,0)}) , \quad \bm{\tilde{M}}_{(1,h,0)} \coloneqq 
    \renewcommand{\arraystretch}{2}
    \begin{bmatrix}
        \displaystyle \frac{1}{nh}\sum_{i=1}^n u_{i,h}^2 K_{i,h} & \displaystyle -\frac{1}{nh}\sum_{i=1}^n u_{i,h} K_{i,h} \\
        \displaystyle -\frac{1}{nh}\sum_{i=1}^n u_{i,h} K_{i,h} & \displaystyle \frac{1}{nh}\sum_{i=1}^n K_{i,h}
    \end{bmatrix}. 
\end{align*}
Using this notation, the second moment can be represented as
\begin{align*}
    \mathbb{E}\left[\hat{f}(0)^2\right] = \mathbb{E}\left[  \left( \frac{1}{nh}\bm{e}_1^\prime |\tilde{\bm{S}}_{(1,h,0)}|^{-1} \tilde{\bm{M}}_{(1,h,0)}\bm{\tilde{\Omega}}_{(1,h,0)} \bm{\hat{F}} \right)^2 \right].
\end{align*}
Next, we bound the inverse of the determinant and the "numerator" part separately.
The determinant $|\tilde{\bm{S}}_{(1,h,0)}|$ is bounded as
\begin{align*}
    |\tilde{\bm{S}}_{(1,h,0)}| &= \left( \frac{1}{nh}\sum_{i=1}^n K_{i,h} \right)\left( \frac{1}{nh}\sum_{i=1}^n u_{i,h}^2 K_{i,h} \right) - \left( \frac{1}{nh}\sum_{i=1}^n u_{i,h} K_{i,h} \right)^2 \\
        &= \frac{1}{n^2h^2} \sum_{i<j} K_{i,h} K_{j,h} (u_{i,h} - u_{j,h})^2 \\
        &= \frac{1}{n^2h^4}\sum_{i<j} K_{i,h} K_{j,h} (X_i - X_j)^2 \ge \frac{K^2(-L/h)}{n^2h^4}  (X_{(1)} - X_{(n)})^2.
\end{align*}
In addition, we can bound the "numerator" component as
\begin{align*}
        & \frac{1}{nh}\bm{e}_1^\prime \tilde{\bm{M}}_{(1,h,0)} \bm{\tilde{\Omega}}_{(1,h,0)} \bm{\hat{F}} \\
        &= \left( \frac{1}{nh^2}\sum_{i=1}^n u_{i,h}K_{i,h}\hat{F}(X_i) \right)\left( \frac{1}{nh}\sum_{i=1}^n K_{i,h} \right) - \left( \frac{1}{nh^2}\sum_{i=1}^n K_{i,h}\hat{F}(X_i) \right)\left( \frac{1}{nh}\sum_{i=1}^n u_{i,h}K_{i,h} \right) \\
        & < \left( \frac{1}{nh^2}\sum_{i=1}^n K_{i,h}\hat{F}(X_i) \right)\left( \frac{1}{nh^2}\sum_{i=1}^n L K_{i,h} \right) <  \frac{L}{h^4} K^2(0).
    \end{align*}
where the inequality follows from $-L/h \le u_{i,h} \le 0$ and the non-negativity of $L$, $h$, $K_{i,h}$ and $\hat{F}(X_i)$, and the final inequality follows from $K_{i,h} \le K(0)$ and $\hat{F}(X_i) \le 1$.
Similarly, we can show that $(nh)^{-1}\bm{e}_1^\prime \tilde{\bm{M}}_{(1,h,0)} \bm{\tilde{\Omega}}_{(1,h,0)} \bm{\hat{F}}  > -{L}/{h^4}K^2(0)$.
Therefore, we obtain that
    \begin{align*}
        \mathbb{E}\left[ \hat{f}^2(0) \right] < \left(\frac{n^2LK^2(0)}{K^2(-L/h)}\right)^2 \mathbb{E}\left[ \frac{1}{(X_{(1)} - X_{(n)})^{4}} \right].
    \end{align*}
Using the same notations as those in the proof of Theorem \ref{theorem: infinite variance}, we can evaluate the expectation as
    \begin{align*}
    &\mathbb{E}\left[ \frac{1}{(X_{(1)} - X_{(n)})^{4}} \right]\\
    &= \int_{-L}^{0}\int_{-L}^{x_{n}}\int_{-L}^{x_{n-1}}\cdots\int_{-L}^{x_{3}}\int_{-L}^{x_{2}} \frac{1}{(x_{1} - x_{n})^4}n! f(x_1)\cdots f(x_n)\, dx_1 dx_2 \ldots dx_{n-1} dx_n\\
    &< n! \Delta^{n} \underbrace{\int_{-L}^{0}\int_{-L}^{x_{n}}\int_{-L}^{x_{n-1}}\cdots\int_{-L}^{x_{3}}\int_{-L}^{x_{2}}  \frac{1}{(x_{1} - x_{n})^4}  \, dx_1 dx_2 \ldots dx_{n-1} dx_n}_{\eqqcolon I_n}.
\end{align*}
We can rewrite $I_n$ by
\begin{align*}
    I_n = \int_{-L}^{0}\int_{-L}^{0}\cdots\int_{-L}^{0}\int_{-L}^{x_{n}} \frac{\mathbf{1}\left\{ x_1\leq x_2\leq\cdots\leq x_{n-1}\leq x_n\right\}}{(x_{1} - x_{n})^4}  \, dx_1 dx_2 \ldots dx_{n-1} dx_n.
\end{align*}
By Fubini's theorem, we can exchange the order of integrations to obtain that
\begin{align*}
    I_n = \int_{-L}^{0}\int_{-L}^{x_{n}}\frac{1}{(n-2)!} (x_n - x_1)^{n-6}\,dx_1 dx_n,
\end{align*}
where we used, for fixed $x_1,x_n\in[-L,0]$ such that $x_1 \leq x_n$,
\begin{align*}
    &\int_{-L}^{0}\cdots\int_{-L}^{0} {\mathbf{1}\left\{ x_1\leq x_2\leq\cdots\leq x_{n-1}\leq x_n\right\}}  \, dx_2 \ldots dx_{n-1}\\
    &=(x_n - x_1)^{n-2} \int_0^1\cdots\int_0^1 \mathbf{1}\left\{ 0\leq y_2\leq\cdots\leq y_{n-1}\leq 1\right\}\, dy_2 \ldots dy_{n-1}\\
    &=(x_n - x_1)^{n-2}\times \frac{1}{(n-2)!}
\end{align*}
and the last equality can be deduced by considering the probability of the $(n-2)$ independent uniformly distributed random variables satisfying the specific order.
Then, we obtain that
\begin{align*}
    I_n = \int_{-L}^{0}\int_{-L}^{x_{n}}\frac{1}{(n-2)!} (x_n - x_1)^{n-6}\,dx_1 dx_n
    =\frac{1}{(n-2)! (n-5) (n-4)} L^{n-4},
\end{align*}
which is bounded if $n\geq6$.
Hence, we can bound the expectation part, which proves the statement.
\end{proof}

\subsection{Corollary}\label{subsec: cor}
\begin{corollary}\label{corl}
    Under the same assumptions for Theorem 1, the local polynomial density estimator $\hat{f}(x_R)$ of degree $p=2$ does not have a second moment.
\end{corollary}

\begin{proof}[Proof of Corollary \ref{corl}]
    We write $\mathbb{E}_{3,\circ}[\cdot]=\mathbb{E}[\cdot  \mid n_0=3,\circ]$.
    Based on a similar idea to Theorem 1, supposing a quadratic polynomial regression with three data points, a straightforward calculation yields
    \begin{align*}
        \mathbb{E}_{3}&\left[\hat{f}(x_R)^2\right]\\
        &=
        \mathbb{E}_{3}\left[\left(\frac{X_{(n)}^2 - 2X_{(n-1)}^2 + X_{(n-2)}^2 - 2x_R\left(X_{(n)} - 2 X_{(n-1)} + X_{(n-2)}\right)}{n\left(X_{(n)}-X_{(n-1)}\right)\left(X_{(n)}-X_{(n-2)}\right)\left(X_{(n-1)}-X_{(n-2)}\right)}\right)^2\right]\\
        &=
        \mathbb{E}_{3}\left[\left(
        \frac{2x_R - \left(X_{(n-1)} + X_{(n-2)}\right)}{n\left(X_{(n)} - X_{(n-2)}\right)\left(X_{(n)} - X_{(n-1)}\right)} - 
        \frac{2x_R - \left(X_{(n)} + X_{(n-1)}\right)}{n\left(X_{(n)} - X_{(n-2)}\right)\left(X_{(n-1)} - X_{(n-2)}\right)}
        \right)^2\right].
    \end{align*}
    Now let $\mathcal{E}$ denote the event $\{X_{(n-2)}\in[x_R-h, x_R-2h/3)\,\mathrm{and}\,X_{(n-1)},X_{(n)}\in[x_R-h/3,x_R)\}$. 
    Note that $\mathbb{P}[n_0=3,\mathcal{E}]>0$, and under $\mathcal{E}$, the first term in the bracket is larger than the second and both are positive.
    Then we have
    \begin{align*}
        \mathbb{E}_{3,\mathcal{E}}&\left[\left(
        \frac{2x_R - \left(X_{(n-1)} + X_{(n-2)}\right)}{n\left(X_{(n)} - X_{(n-2)}\right)\left(X_{(n)} - X_{(n-1)}\right)} - 
        \frac{2x_R - \left(X_{(n)} + X_{(n-1)}\right)}{n\left(X_{(n)} - X_{(n-2)}\right)\left(X_{(n-1)} - X_{(n-2)}\right)}
        \right)^2\right]\\
        \geq&
        \mathbb{E}_{3,\mathcal{E}}\left[\left(
        \frac{2x_R - \left(X_{(n-1)} + X_{(n-2)}\right)}{nh\left(X_{(n)} - X_{(n-1)}\right)} - 
        \frac{2x_R - \left(X_{(n)} + X_{(n-1)}\right)}{nh\left(X_{(n-1)} - X_{(n-2)}\right)}
        \right)^2\right]\\
        \geq&
        \mathbb{E}_{3,\mathcal{E}}\left[\left(
        \frac{x_R - X_{(n-2)}}{nh\left(X_{(n)} - X_{(n-1)}\right)} - 
        \frac{x_R - X_{(n)}}{nh\left(X_{(n-1)} - X_{(n-2)}\right)}
        \right)^2\right]\\
        \geq&
        \mathbb{E}_{3,\mathcal{E}}\left[\left(
        \frac{x_R - X_{(n-2)}}{nh\left(x_R - X_{(n-1)}\right)} - 
        \frac{h/3}{nh\left(h/3\right)}
        \right)^2\right]\\
        \geq&
        \mathbb{E}_{3,\mathcal{E}}\left[\left(
        \frac{h/3}{nh\left(x_R - X_{(n-1)}\right)} - 
        \frac{1}{nh}
        \right)^2\right].
    \end{align*}
    The remainder is the same as the proof of Theorem 1. 
\end{proof}

\subsection{Local Linear Regression}\label{subsec: local linear}
We here explain the relation between the local polynomial regression and the LPD. Write the regression function by $r(x)$, its local linear estimator by $\hat{r}(x)$, and the density of the regressor by $f(x)$. Assume for simplicity that the error satisfies $\V{\varepsilon_i}=\sigma^2$, and $[x_L,x_R)=[0,1)$. We write some bounded constants by $C_j$ below.

Utilizing the homoskedasticity assumption and the law of total variance repeatedly, \cite{Seifert_Gasser:1996} showed that 
\begin{align*}
    \V{\hat{r}(x_R)}\geq C_1 \times 
    \mathbb{V}_{\mathrm{U},r\equiv0,f\equiv1}\left[\hat{r}(x)|n_0=2\right],
\end{align*}
where $\mathbb{V}_{\mathrm{U},r\equiv0,f\equiv1}$ is the variance when using the uniform kernel, $r(x)\equiv0$, and $f(x)\equiv1$.
Although \cite{Seifert_Gasser:1996} immediately concludes that the last term equals infinity, an important point would be in the omitted part, as it will show what happens and clarify what the connection is. Now, note that we can write
\begin{align}
    &\mathbb{V}_{\mathrm{U},r\equiv0,f\equiv1}\left[\hat{r}(x_R) \mid n_0=2\right]
    = 
    \mathbb{E}_{2,r\equiv0,f\equiv1}\left[\left(\frac{Y_{(n)}-Y_{(n-1)}}{X_{(n)}-X_{(n-1)}}\right)^2\right] - C_2\label{fit line LLR}\\
    &\quad=
    \mathbb{E}_{2,r\equiv0,f\equiv1}\left[\left(\frac{\varepsilon_{(n)}-\varepsilon_{(n-1)}}{X_{(n)}-X_{(n-1)}}\right)^2\right] - C_2
    =
    \mathbb{E}_{2,r\equiv0,f\equiv1}\left[\frac{2\sigma^2}{\left(X_{(n)}-X_{(n-1)}\right)^2}\right] - C_2,\label{similar form LLR}
\end{align}
where $(Y_{(n)},\varepsilon_{(n)})$ is $(Y,\varepsilon)$ that corresponds to $X_{(n)}$.
We can show that this is infinite by following similar steps as we take in our proof since (\ref{similar form LLR}) is essentially the same as (\ref{similar form LPD}).
Now we can see that the LPD and the local linear regression share a similar structure through (\ref{similar form LPD}) and (\ref{similar form LLR}), which are a direct consequence of (\ref{fit line LPD}) and (\ref{fit line LLR}). 
This fact suggests that the fundamental factor for the variance property is the same: Roughly speaking, when the local sample size is small relative to the polynomial degree, the local polynomial overfits the data, and in such a case, the estimate or the slope can take an arbitrarily large value with a certain probability, since the data point can be located close enough.

\subsection{Proofs of Lemmas}
\begin{proof}[Proof of Lemma \ref{lemma:gaussian approximation}]
We derive a bound on the normal approximation error of $T_3$ under $\mathbb{H}_{1,\mathtt{LA}}$ and then show the convergence of the bound. 
For each $n$, probabilities are taken under $\mathbb{P}_n$.
Convergence statements are considered for the sequence $\{\mathbb{P}_n\}_{n\in\mathbb{N}_+}$ as $n\to \infty$.
Since Assumption \ref{as:DGP-local} imposes the smoothness and boundedness conditions on the DGP uniformly in $n$, under Assumptions \ref{as:DGP-local} and \ref{as:kernel}, we can use the same convergence rates of the remainder terms, as given in Lemma 7-10 of the supplemental appendix for \cite{Cattaneo_etal:2020}. 
Similarly, from Theorem 2 in the supplemental appendix of \cite{Cattaneo_etal:2020}, we can see that the variance is consistent under Assumption \ref{as:DGP-local} and \ref{as:kernel}.
Therefore, it holds that
\begin{align*}
    T_3 &= \frac{\tau_n + \bm{e}_1^\top \left[ \bm{H}_+^{-1}\left(f_{n}(0_+)\bm{S}_{(3,0_+)}\right)^{-1} \hat{\mathbf{L}}_{n,+} - \bm{H}^{-1}_-\left(f_{n}(0_-)\bm{S}_{(3,0_-)}\right)^{-1} \hat{\mathbf{L}}_{n,-} \right] + r_{n} }{\sqrt{\frac{1}{nh_{2+}^{\mathtt{MSE}}} \mathscr{V}_{n,+} + \frac{1}{nh_{2-}^{\mathtt{MSE}}} \mathscr{V}_{n,-}}},
\end{align*}
where $r_{n}$ is a $n$-dependent sequence satisfying $ (\frac{1}{nh_{2+}^{\mathtt{MSE}}} \mathscr{V}_{n,+} + \frac{1}{nh_{2-}^{\mathtt{MSE}}} \mathscr{V}_{n,-})^{-1/2} |r_{n}| \to 0$ as $n \to \infty$,
\begin{align*}
    \bm{S}_{(3,0_+)}&\coloneqq \int_0^\infty \bm{r}_3(u) \bm{r}_3(u)^\top K(u)\,du,\quad 
    \bm{H}_+ \coloneqq \mathrm{diag}\left[1, h_{2+}^{\mathtt{MSE}}, \dots, (h_{2+}^{\mathtt{MSE}})^p\right],\\
    \hat{\mathbf{L}}_{n,+}&\coloneqq \int_0^{s_r/h_{2+}^{\mathtt{MSE}}} \bm r_3\left( u\right) K\left( u\right)\{\hat{F}_{n}(uh_{2+}^{\mathtt{MSE}}) - F_{n}(uh_{2+}^{\mathtt{MSE}})\}f_{n}(uh_{2+}^{\mathtt{MSE}}) \,du,
\end{align*}
and $\bm{S}_{(3,0_-)}$ and $\hat{\mathbf{L}}_{n,-}$ are defined similarly.
Observe that the leading term is the sample mean of i.i.d.~random variables:
\begin{align*}
    \hat{L}_{n} &\coloneqq \bm{e}_1^\top \left[ \bm{H}_+^{-1}\left(f_{n}(0_+)\bm{S}_{(3,0_+)}\right)^{-1} \hat{\mathbf{L}}_{n,+} - \bm{H}^{-1}_-\left(f_{n}(0_-)\bm{S}_{(3,0_-)}\right)^{-1} \hat{\mathbf{L}}_{n,-} \right] \\
    & = \frac{1}{h_{2+}^{\mathtt{MSE}}}\int_0^{s_r/h_{2+}^{\mathtt{MSE}}} \frac{1}{f_{n}(0_+)}\bm{e}_1^\top\bm{S}^{-1}_{(3,0_+)} \bm r_3\left( u\right) K\left( u\right)\{\hat{F}_{n}(uh_{2+}^{\mathtt{MSE}}) - F_{n}(uh_{2+}^{\mathtt{MSE}})\}f_{n}(uh_{2+}^{\mathtt{MSE}}) \,du \\
    & \quad - \frac{1}{h_{2-}^{\mathtt{MSE}}}\int^0_{-s_\ell/h_{2-}^{\mathtt{MSE}}} \frac{1}{f_{n}(0_-)}\bm{e}_1^\top\bm{S}^{-1}_{(3,0_-)} \bm r_3\left( u\right) K\left( u\right)\{\hat{F}_{n}(uh_{2-}^{\mathtt{MSE}}) - F_{n}(uh_{2-}^{\mathtt{MSE}})\}f_{n}(uh_{2-}^{\mathtt{MSE}}) \,du \\
    & = \frac{1}{nh_{2+}^{\mathtt{MSE}}}\sum_{i=1}^n \int_0^{s_r/h_{2+}^{\mathtt{MSE}}} \frac{1}{f_{n}(0_+)}\bm{e}_1^\top\bm{S}^{-1}_{(3,0_+)} \bm r_3\left( u\right) K\left( u\right)\{\bm{1}[X_i\le uh_{2+}^{\mathtt{MSE}}] - F_{n}(uh_{2+}^{\mathtt{MSE}})\}f_{n}(uh_{2+}^{\mathtt{MSE}}) \,du \\
    & \quad -  \frac{1}{nh_{2-}^{\mathtt{MSE}}}\sum_{i=1}^n \int^0_{-s_\ell/h_{2-}^{\mathtt{MSE}}} \frac{1}{f_{n}(0_-)}\bm{e}_1^\top\bm{S}^{-1}_{(3,0_-)} \bm r_3\left( u\right) K\left( u\right)\{\bm{1}[X_i\le uh_{2-}^{\mathtt{MSE}}] - F_{n}(uh_{2-}^{\mathtt{MSE}})\}f_{n}(uh_{2-}^{\mathtt{MSE}}) \,du.
\end{align*}
Let $\hat{Z}_{n} \coloneqq (\frac{1}{nh_{2+}^{\mathtt{MSE}}} \mathscr{V}_{n,+} + \frac{1}{nh_{2-}^{\mathtt{MSE}}} \mathscr{V}_{n,-})^{-1/2} \hat{L}_{n}$.
Then, the Berry-Esseen theorem for the sample mean of i.i.d. random variables implies that
\begin{align*}
    \sup_{z\in\mathbb{R}} \left| \mathbb{P}_{n}\left(\hat{Z}_{n} \le z\right) - \Phi(z) \right| & \le C_{\text{BE}} \frac{\mathbb{E}_{n}[|Y_{i,n}|^3]}{\sqrt{n} \{\mathbb{V}_{n}[Y_{i,n}]\}^{3/2}},
\end{align*}
where $C_{\text{BE}}$ is a constant independent of $n$, and $Y_{i,n}$ is defined as
\begin{align*}
    Y_{i,n} & \coloneqq \frac{1}{h_{2+}^{\mathtt{MSE}}}\int_0^{s_r/h_{2+}^{\mathtt{MSE}}} \frac{1}{f_{n}(0_+)}\bm{e}_1^\top\bm{S}^{-1}_{(3,0_+)} \bm r_3\left( u\right) K\left( u\right)\{\bm{1}[X_i\le uh_{2+}^{\mathtt{MSE}}] - F_{n}(uh_{2+}^{\mathtt{MSE}})\}f_{n}(uh_{2+}^{\mathtt{MSE}}) \,du \\
    & \quad - \frac{1}{h_{2-}^{\mathtt{MSE}}}\int^0_{-s_\ell/h_{2-}^{\mathtt{MSE}}} \frac{1}{f_{n}(0_-)}\bm{e}_1^\top\bm{S}^{-1}_{(3,0_-)} \bm r_3\left( u\right) K\left( u\right)\{\bm{1}[X_i\le uh_{2-}^{\mathtt{MSE}}] - F_{n}(uh_{2-}^{\mathtt{MSE}})\}f_{n}(uh_{2-}^{\mathtt{MSE}}) \,du.
\end{align*}
Here, observe that
\begin{align*}
    & \sup_{z\in\mathbb{R}} \left| \mathbb{P}_{n}\left(T_3 - \tilde{\tau}_n  \le z\right) - \Phi(z) \right| = \sup_{z\in\mathbb{R}} \left| \mathbb{P}_{n}\left(T_3   \le z \right) - \Phi(z -  \tilde{\tau}_n) \right|,
\end{align*}
and that
\begin{align*}
     \sup_{z\in\mathbb{R}} \left| \mathbb{P}_{n}\left(T_3 - \tilde{\tau}_n  \le z\right) - \Phi(z) \right| \le  \sup_{z\in\mathbb{R}} \left| \mathbb{P}_{n}\left(\hat{Z}_{n} \le z\right) - \Phi(z) \right|  + o(1).
\end{align*}
Then, we obtain that
\begin{align*}
    \sup_{z\in\mathbb{R}}  \left| \mathbb{P}_{n}(T_3 \le z) - \Phi(z -  \tilde{\tau}_n) \right| \lesssim  \frac{\mathbb{E}_{n}[|Y_{i,n}|^3]}{\sqrt{n} \{\mathbb{V}_{n}[Y_{i,n}]\}^{3/2}}.
\end{align*}
Hence, it suffices to show that this bound converges to $0$ as $n\to \infty$.
A similar evaluation to Lemma 9 in the supplemental appendix of \cite{Cattaneo_etal:2020} shows that $\{\mathbb{V}_{n}[Y_{i,n}]\}^{3/2} \asymp  \{({h_{2+}^{\mathtt{MSE}}})^{-1} + ({h_{2-}^{\mathtt{MSE}}})^{-1}\}^{3/2}$.
Also, since $\sup_{t\in\mathbb{R}}|\bm{1}[X_i\le t] - F_{n}(t)\}| \le 1$, it holds that
\begin{align*}
    |Y_{i,n}| 
    & \le \left(\frac{1}{{h_{2+}^{\mathtt{MSE}}}} \|\bm{e}_1^\top \bm{S}^{-1}_{(3,0_+)} \| \int_0^{s_r/h_{2+}^{\mathtt{MSE}}} \bm \|\bm r_3\left( u\right)\||K(u)|\,du \right)\left( \sup_{x\in\mathcal{X}}|f_{n}(x)|   \right) \\
    & \quad + \left(\frac{1}{{h_{2-}^{\mathtt{MSE}}}} \|\bm{e}_1^\top \bm{S}^{-1}_{(3,0_-)} \| \int^0_{-s_\ell/h_{2-}^{\mathtt{MSE}}} \bm \|\bm r_3\left( u\right)\||K(u)|\,du \right)\left( \sup_{x\in\mathcal{X}}|f_{n}(x)|   \right),
\end{align*}
where $\|\cdot\|$ is the Euclidean norm.
Therefore, under Assumptions \ref{as:DGP-local} and \ref{as:kernel}, we have $|Y_{i,n}|^3 \lesssim  \{({h_{2+}^{\mathtt{MSE}}})^{-1} + ({h_{2-}^{\mathtt{MSE}}})^{-1}\}^3$.
Since $h_{2+}^{\mathtt{MSE}}   \asymp n^{-1/5}$ and $h_{2-}^{\mathtt{MSE}}   \asymp n^{-1/5}$, we conclude that
\begin{align*}
     \frac{\mathbb{E}_{n}[|Y_{i,n}|^3]}{\sqrt{n} \{\mathbb{V}_{n}[Y_{i,n}]\}^{3/2}} \lesssim \frac{ \{({h_{2+}^{\mathtt{MSE}}})^{-1} + ({h_{2-}^{\mathtt{MSE}}})^{-1}\}^{3} }{\sqrt{n} \{({h_{2+}^{\mathtt{MSE}}})^{-1} + ({h_{2-}^{\mathtt{MSE}}})^{-1}\}^{3/2}} \asymp n^{-1/2 } n^{3/10}  \to 0.
\end{align*}
Summing up, $\sup_{z\in\mathbb{R}} \left| \mathbb{P}_{n}(T_3 \le z) - \Phi(z -  \tilde{\tau}_n) \right| \to 0$.
\end{proof}

\begin{proof}[Proof of Lemma \ref{lemma: equiv kernel}]
    The first statement follows by $K^*_{p,K}(0)=\int_{0}^{0} \bm{e}_1^\prime \bm{A}_{p,K}^{-1} \bm{r}_p\left( z\right) K\left( z\right)=0$.
    We show the second statement below.
    Notice that $\mathscr{K}_{p,K}(z) \coloneqq \bm{e}_1^\prime \bm{A}_{p,K}^{-1} \bm{r}_p\left( z\right) K\left( z\right)$ is the usual equivalent kernel of local polynomial fitting for estimating the first derivative of regression function (\citealp[p.70]{Fan_Gijbels:1996}; $K^*_{\nu,c}$ with $\nu=1,c=0$ in their notation).
    Therefore, $\mathscr{K}_{p,K}$ satisfies
    \begin{align}
        \int_{-\infty}^{0} z^j \mathscr{K}_{p,K}(z) \,dz = \delta_{1,j}\,\,(0\leq j \leq p),\label{moments of equiv kernel in LPR}
    \end{align}
    where $\delta_{1,j}$ takes $1$ if $j=1$ and $0$ otherwise.
    From the standard property of the integration and \eqref{moments of equiv kernel in LPR} with $j=0$,
     \begin{align*}
            \int_{-\infty}^{0} \mathscr{K}_{p,K}(z)
            \mathbf{1}\left\{u \leq z\right\}
            \,dz 
            &= 
            \int_{-\infty}^{0} \mathscr{K}_{p,K}(z)\,dz - \int_{-\infty}^{u} \mathscr{K}_{p,K}(z)\,dz \\
            &= 
            - \int_{-\infty}^{0} \mathscr{K}_{p,K}(z) \mathbf{1}\left\{z \leq u\right\}\,dz.
    \end{align*}
    Thus, we have
    \begin{align}
        \int_{-\infty}^{0} u^j K^*_{p,K}(u) \,du
            &= 
            \int_{-\infty}^{0} u^j \left\{\int_{-\infty}^{0} \mathscr{K}_{p,K}(z)
            \mathbf{1}\left\{u \leq z\right\}
            \,dz \right\}\,du \nonumber\\
            & = \int_{-\infty}^{0} u^j \left\{- \int_{-\infty}^{0} \mathscr{K}_{p,K}(z)
            \mathbf{1}\left\{u \geq z\right\}
            \,dz \right\}\,du. \label{eq:exchange}
    \end{align}
    In addition, by a simple evaluation, we can see that
    \begin{align*}
        \int_{-\infty}^{0}  \int_{-\infty}^{0} \bigg|u^j \mathscr{K}_{p,K}(z)
            \mathbf{1}\left\{u \geq z\right\}\bigg|
            \,du \,dz \le  \frac{1}{j+1}\int_{-\infty}^{0} \sum_{k=0}^{p} |a_k| |z|^{k+j+1} K\left( z\right) \,dz,
    \end{align*}
    where $a_k$'s are some constants.
    Together with the assumption $\int_0^{\infty} u^{2p}K(u)\,du<\infty$, the order of the integration in \eqref{eq:exchange} is exchangeable so that we have
    \begin{align*}
        \int_{-\infty}^{0} u^j K^*_{p,K}(u) \,du
            &= 
            -\int_{-\infty}^{0} \mathscr{K}_{p,K}(z) \int_{z}^{0} 
            u^j
            \,du \,dz\\
            &=
            \frac{1}{j+1}\int_{-\infty}^{0} \mathscr{K}_{p,K}(z) z^{j+1} \,dz = 
            \frac{\delta_{0,j}}{j+1}=\delta_{0,j}\,\,(0\leq j \leq p-1).
    \end{align*}
    This proves the moment conditions.
\end{proof}

\subsection{Infeasible Optimal Weighting Scheme}\label{subsec: infeasible optimal}
In this section, we derive the infeasible optimal weight.

\subsubsection{Some Preliminaries}
Throughout this section, we assume the following condition:
\begin{assumption}
    Kernel function $K$ is of second-order, nonnegative, symmetric, absolutely continuous on $(-\infty,0]$, $K(0)>0$, $\int_{0}^{\infty} u^{2p} K(u) \,du < \infty$, and
    \begin{align}
        \sup_{u\leq0} (1+|u|^p) K(u) \leq C_0,\quad \text{ and }\,
        \esssup_{u\leq0} (1+|u|^{p+1}) |K^\prime(u)| \leq C_1,\label{eq:cond ker}
    \end{align}
    where $C_0,C_1$ are some positive constants and $K^\prime$ is the derivative of $K$ (which exists almost everywhere).
    The class of kernel functions satisfying these conditions will be denoted by $\mathcal{K}$.
\end{assumption}

All the conditions are fairly standard.
In particular, compactly supported $C^1$ kernel functions satisfy \eqref{eq:cond ker}.
The Gaussian or Laplace kernels also meet the assumption.

We will see some properties of $K_{p,K}^*$ when $K\in\mathcal{K}$.
We first show $\mathscr{K}_{p,K}(z) >0$ holds near $z=0$.
Here, write $\bm{A}_{p,K} = \mathrm{diag}(1,-1,\ldots)\bm{D}_{p,K}\mathrm{diag}(1,-1,\ldots)$, where $\bm{D}_{p,K}\coloneqq\int_{0}^{\infty} \bm{r}_p(u)\bm{r}_p(u)^\prime K(u)\, du$. 
Then, \citet[Theorem 4.4]{Pinkus:2009} shows that $\bm{D}_{p,K}$ is strictly totally positive, and hence $\bm{A}_{p,K}^{-1}$ is also totally positive (\citealp[Proposition 1.6]{Pinkus:2009}), implying that all elements of $\bm{A}_{p,K}^{-1}$ is positive.
Hence, $\bm{e}_1^\prime \bm{A}_{p,K}^{-1} \bm{r}_p$ is positively valued over $z\in[-\varepsilon,0]$ for sufficiently small $\varepsilon>0$. Combined with the assumption $K\geq0$ and $K(0)>0$, $\mathscr{K}_{p,K}(z)$ is positively valued just below the origin.
Note that this implies that $K_{p, K}^*$ is also positive near the origin, since $K_{p, K}^*(u) = \int_u^0 \mathscr{K}_{p, K}(z)\,dz$ by definition.

Next, we discuss the sign change property of the equivalent kernel.
The moment condition in \eqref{moments of equiv kernel of LPD} implies $K_{p, K}^*$ changes its sign at least $p-1$ times. At the same time, $\mathscr{K}_{p, K}$ changes its sign up to $p$ times. Combined with the fact that $\int_{-\infty}^{0}\mathscr{K}_{p, K}=0$ (see \eqref{moments of equiv kernel in LPR}), we can see that $K_{p, K}^*(u) = \int_u^0 \mathscr{K}_{p, K}(z)\,dz$ changes its sign up to $p-1$ times. To sum up, $K_{p, K}^*$ changes its sign $p-1$ times.
Then, we can take zeros $u_j,j=1,\ldots,p-1$ and define $q(u) \coloneqq \prod_{j=1}^{p-1}(u - u_j)$.

With this $q$, we define $\tilde{K}(u) \coloneqq {K_{p, K}^*(u)}/(C_{p,K} \cdot q(u))$, where $C_{p,K}$ is a normalizing constant so that $\int_{-\infty}^{0}\tilde{K}=1$.
Note that the sign of $K_{p, K}^*(u)$ and $q(u)$ coinsides, since $K_{p, K}^*(u)$ changes its sign exactly $p-1$ times and $K_{p, K}^*(u)>0$ near the origin.
Hence, $\tilde{K}(u)\geq0$. 
Further, $\tilde{K}(u)$ is Lipschitz continuous.
To see this, note that $u^j K(u)$ is bounded and Lipschitz for $j=0,\ldots,p$. Then, $\mathscr{K}_{p, K}$ is also bounded and Lipschitz, implying that so is $K_{p,K}^*$.
Define 
\begin{align*}
    k_j(u)\coloneqq \begin{cases}
        \frac{K_{p,K}^*(u)}{u-u_j} & \text{ if } u\neq u_j,\\
        -\mathscr{K}_{p, K}(u_j) & \text{ if } u= u_j,
    \end{cases}
\end{align*}
which is bounded and Lipschitz by noting that 
\begin{align*}
    k_j(u) = \frac{1}{u - u_j} \left(-\int_{u_j}^u \mathscr{K}_{p, K}(z)\,dz\right)
    = - \int_0^1 \mathscr{K}_{p, K}(u_j+t(u-u_j))\,dt
\end{align*}
and that $|k_j(u)-k_j(v)|\lesssim \int_0^1 t|u-v|\,dt = |u-v|/2$.
Then, in some neighborhood $[u_j-\delta_j, u_j+\delta_j], \exists\delta_j>0$, $\tilde{K}(u) = {K_{p, K}^*(u)}/(C_{p,K} \cdot q(u)) = k_j(u)/(C_{p,K} \cdot q_{-j}(u))$, where $q(u) = (u-u_j)q_{-j}(u)$, and this is (locally) Lipschitz continuous.
Based on this, the global Lipschitz continuity follows easily.

Now, rewrite $q(u) = \bm{\beta}^\top \bm{r}_{p-1}(u)$.
Then, the moment condition can be rewritten as 
\begin{align*}
    (1,0,\ldots,0) &= C_{p,K}\bm{\beta}^\top \int_{-\infty}^0 \bm{r}_{p-1}(u) \bm{r}_{p-1}(u)^\top \tilde{K}(u)\,du
    = C_{p,K}\bm{\beta}^\top \bm{A}_{p-1, \tilde{K}},
\end{align*}
that is $\bm{\beta}^\top =(1,0,\ldots,0)\bm{A}_{p-1, \tilde{K}}^{-1}/{C_{p,K}}$
Hence, we have that
\begin{align*}
    K_{p, K}^*(u) &= {C_{p,K}} \tilde{K}(u)q(u) = {C_{p,K}} \tilde{K}(u) \bm{\beta}^\top \bm{r}_{p-1}(u)\\
    &= (1,0,\ldots,0)\bm{A}_{p-1, \tilde{K}}^{-1}\bm{r}_{p-1}(u)\tilde{K}(u)\\
    &= \mathscr{K}_{p-1, \tilde{K}}(z).
\end{align*}

\subsubsection{Optimal Weight}
We are now interested in the minimization of the asymptotic MSE.
By the standard local polynomial regression theory, we shall consider
\begin{align*}
    \min_{K\in\mathcal{K}}\left|\int_{-\infty}^0 u^{p} K_{p, K}^*(u)\,du\right|\left\{\int_{-\infty}^0 K_{p, K}^*(u)^2\,du\right\}^p.
\end{align*}
We have seen in the previous section that the following problem is more general than this problem:
\begin{align*}
    \min_{\bar{K}:\bar{K}\geq0\text{ and Lipschitz}}\left|\int_{-\infty}^0 u^{p} \mathscr{K}_{p-1, \tilde{K}}(u)\,du\right|\left\{\int_{-\infty}^0 \mathscr{K}_{p-1, \tilde{K}}(u)^2\,du\right\}^p.
\end{align*}
This problem is solved by \citet[Theorem 2]{Cheng_etal:1997}, which shows that the equivalent kernel of the triangular kernel, say $\mathscr{K}_{p-1, T}$, is the solution.
Hence, this is the optimal weighting scheme within an enriched class of weight functions. However, it is not feasible in the LPD estimation since $\mathscr{K}_{p-1, T}(0)\neq0$.

\section{Equivalent Kernels}\label{sec: oa equiv}
We outline the derivation of the equivalent kernel.
For notational simplicity, we write 
\begin{align*}
    & u_{i,h} \coloneqq \frac{X_i-x}{h}, 
    \quad K_{i,h} \coloneqq K( u_{i,h} ), 
    \quad \mathbf{K}_h \coloneqq \mathrm{diag}[ K_{1,h}, \dots, K_{n,h} ],  \\ 
    &\mathbf{X}_h \coloneqq \left[ u_{i,h}^j \right]_{1\le i\le n,~ 0\le j \le p},
    \quad \bm{H} \coloneqq \mathrm{diag}[1, h, \dots, h^p].
\end{align*}
In addition, we define
\begin{align*}
    & \tilde{\bm{S}}_{(p,h,x)} \coloneqq \frac{1}{nh} \mathbf{X}_h^\prime \mathbf{K}_h \mathbf{X}_h = \frac{1}{nh} \sum_{i=1}^n \bm{r}_p(u_{i,h})\bm{r}_p(u_{i,h})^\prime K_{i,h}, \\
    & \bm{S}_{(p,x)} \coloneqq \int_{-\infty}^{(x_R-x)/h} \bm{r}_p(u)\bm{r}_p(u)^\prime K(u)\,du.
\end{align*}
For the analysis of the boundary point $x_R$, we introduce
\begin{align*}
    \bm{A}_{p,K} &= \int_{-\infty}^{0} \bm{r}_p(u)\bm{r}_p(u)^\prime K(u) \,du. 
\end{align*}
Note that $ \bm{S}_{(p,x_R)} = \bm{A}_{p,K}$.
Finally, for the finite sample analysis below, we define
\begin{align*}
    & \bm{\hat{F}} \coloneqq (\hat{F}(X_1), \dots, \hat{F}(X_n))^\prime, \\
    & \bm{\tilde{\Omega}}_{(p,h,K)} \coloneqq
    \frac{1}{h} \mathbf{X}_h^\prime \mathbf{K}_h =
    \frac{1}{h} 
    \begin{bmatrix} K_{1,h}\bm{r}_p(u_{1,h})& \dots & K_{n,h}\bm{r}_p(u_{n,h}) \end{bmatrix}.
\end{align*}
Then, the LPD estimator has the following closed form:
\begin{align*}
    \hat{f}(x) = \bm{e}_1^\prime \bm{H}^{-1} \tilde{\bm{S}}_{(p,h,x)}^{-1} \frac{1}{nh} \sum_{i=1}^{n} \bm{r}_p\left( u_{i,h} \right) K\left( u_{i,h} \right) \hat{F}(X_i).
\end{align*}

In the following, we assume that (i) $h\to 0$ and $nh\to\infty$ as $n\to \infty$, (ii)$K(\cdot)$ is a non-negative, symmetric kernel function such that $\int K(u)du =1$, and (iii) density function $f$ satisfies standard regularity conditions.

From Lemma 1 of the Supplemental Appendix of \cite{Cattaneo_etal:2020} and the continuous mapping theorem, it follows that
\begin{align}
    \tilde{\bm{S}}_{(p,h,x)}^{-1} = \frac{1}{f(x)}\bm{S}_{(p,x)}^{-1} + o_p(1). \label{eq:inv_Gamma}
\end{align}
In addition, from simple algebra, we can see that
\begin{align}
    & \frac{1}{nh} \sum_{i=1}^{n} \bm{r}_p\left( u_{i,h} \right) K\left( u_{i,h} \right) \hat{F}(X_i) \nonumber \\
    & = \int_{-\infty}^{(x_R-x)/h} \bm{r}_p(u) K(u)\hat{F}(x+uh)f(x+uh)du \nonumber\\
    & \quad + \frac{1}{nh} \sum_{i=1}^n \bm{r}_p\left( u_{i,h} \right) K\left(u_{i,h}\right) \left\{\hat{F}(X_i) - F(X_i) \right\} \nonumber\\
    & \quad\quad - \int_{-\infty}^{(x_R-x)/h} \bm{r}_p(u) K(u)\left\{\hat{F}(x+uh)-F(x+uh)\right\}f(x+uh)du \nonumber\\
    & \quad + \frac{1}{nh} \sum_{i=1}^{n} \bm{r}_p\left(u_{i,h}\right) K\left(u_{i,h}\right) F(X_i) - \int_{-\infty}^{(x_R-x)/h} \bm{r}_p(u) K(u)F(x+uh)f(x+uh)du \nonumber\\
    & = \int_{-\infty}^{(x_R-x)/h} \bm{r}_p(u) K(u)\hat{F}(x+uh)f(x+uh)du  \nonumber\\
    & \quad + \frac{1}{n^2h} \sum_{i=1}^n \bm{r}_p\left(u_{i,h}\right) K\left(u_{i,h}\right) \left( 1 - F(X_i)\right) \nonumber\\
    & \quad + \frac{1}{n^2h} \sum_{i\neq j} \Bigg\{ \bm{r}_p\left(u_{i,h}\right)\Big(\mathbf{1}[X_j\le X_i] - F(X_i) \Big)K\left(u_{i,h}\right) \nonumber\\
    & \quad\quad - \mathbb{E}\left[ \bm{r}_p\left(u_{i,h}\right)\Big(\mathbf{1}[X_j\le X_i] - F(X_i) \Big)K\left(u_{i,h}\right) \mid X_j \right] \Bigg\} \nonumber\\
    & \quad + \frac{1}{nh} \sum_{i=1}^{n} \bm{r}_p\left(u_{i,h}\right) K\left(u_{i,h}\right) F(X_i) - \int_{-\infty}^{(x_R-x)/h} \bm{r}_p(u) K(u)F(x+uh)f(x+uh)du 
    \label{eq:lpd_expansion} \\
    & = \int_{-\infty}^{(x_R-x)/h} \bm{r}_p(u) K(u)\hat{F}(x+uh)f(x+uh)du + \hat{\mathbf{B}}_{\mathtt{LI}} + \hat{\mathbf{R}} + \hat{\mathbf{W}}, \nonumber
\end{align}
where $\hat{\mathbf{B}}_{\mathtt{LI}}$, $\hat{\mathbf{R}}$ and $\hat{\mathbf{W}}$ are the second, the third and the final terms in \eqref{eq:lpd_expansion} respectively.
Lemma 2 and 4 of Supplemental Appendix of \cite{Cattaneo_etal:2020} state that $\hat{\mathbf{B}}_{\mathtt{LI}} = o_p(1)$ and $\hat{\mathbf{R}} = o_p(1)$. 
The convergence of $\hat{\mathbf{W}}$ is obvious from the law of large numbers.
So it holds that
\begin{align}
     & \frac{1}{n} \sum_{i=1}^{n} \bm{r}_p\left(u_{i,h}\right) K\left(u_{i,h}\right) \hat{F}(X_i) \nonumber\\
     & = \int_{-\infty}^{(x_R-x)/h} \bm{r}_p(u) K(u)\hat{F}(x+uh)f(x+uh)du + o_p(1). \label{eq:nu_leading}
\end{align}
By \eqref{eq:inv_Gamma}, \eqref{eq:nu_leading}, and $\bm{e}_1^\prime\bm{H}^{-1} = \frac{1}{h}\bm{e}_1^\prime$, we have
\begin{align*}
    \hat{f}(x)
    &=
    \frac{1}{h}\bm{e}_1^\prime \left\{\frac{1}{f(x)} \bm{S}_{(p,x)}^{-1}\right\}
    \int_{-\infty}^{(x_R-x)/h}   r_p\left( u\right) K\left( u\right)\hat{F}(x+uh)f(x+uh) du + o_p(1)\\
    &=
    \frac{1}{nh}\sum_{i=1}^{n} 
    \int_{-\infty}^{(x_R-x)/h} \bm{e}_1^\prime \bm{S}_{(p,x)}^{-1} \bm{r}_p\left( u\right) K\left( u\right)
    \mathbf{1}\left\{u_{i,h} \leq u\right\}
    du + o_p(1).
\end{align*}
That is, the equivalent kernel $K^*$ is given by
\begin{align*}
    K^*(v) = 
    \int_{-\infty}^{(x_R-x)/h} \bm{e}_1^\prime \bm{S}_{(p,x)}^{-1} \bm{r}_p\left( u\right) K\left( u\right)
    \mathbf{1}\left\{v \leq u\right\}
    du.
\end{align*}

\section{Asymptotic Efficiency at Interior Points}\label{sec: oa int}
We compute the asymptotic efficiency at interior points for $p=2$.
Table \ref{tab: kernel asy efficiency int} summarizes the result.
\begin{table}[t]
    \begin{center}
    \scalebox{1}{
    \begin{tabular}{l|ccc}
    \hline\hline
        Kernel Function & $\mathcal{V}_{2,K}$ & $\mathcal{Q}_{2,K}$ & $\Theta_{K}$\\ \hline
        \multirow{2}*{Uniform} & 0.600 & 0.542 & 1.129\\ 
         & (3.84) & (0.89) & (1.25)\\ \hline
        \multirow{2}*{Epanechnikov} & 0.714 & 0.544 & 1.073\\ 
         & (4.57) & (0.89) & (1.18)\\ \hline
         \multirow{2}*{Biweight} & 0.816 & 0.548 & 1.044\\ 
         & (5.22) & (0.89) & (1.15)\\ \hline
         \multirow{2}*{Gaussian} &  0.282 & 0.564 & 0.951\\
         & (1.81) & (0.92) & (1.05)\\ \hline
         \multirow{2}*{Triangular} & 0.743 & 0.546 & 1.068\\ 
         & (4.75) & (0.89) & (1.18)\\ \hline
         \multirow{2}*{$2$-Triangular} & 0.893 & 0.553 & 1.034\\ 
         & (5.71) & (0.90) & (1.14)\\ \hline
         \multirow{2}*{$3$-Triangular} & 1.045 &0.560 & 1.012\\ 
         & (6.69) & (0.91) & (1.12)\\ \hline
        \multirow{2}*{Laplace} & 0.156 & 0.612 & 0.906\\ 
        & (1.00) & (1.00) & (1.00)\\ \hline
    \end{tabular}}
    \caption{Asymptotic Variance and Efficiency Relative to the Laplace Kernel (Interior Case)}
    \label{tab: kernel asy efficiency int}
    \end{center}
\end{table}
At interior points, the Laplace kernel performs worst in terms of MSE, while the uniform kernel delivers the best performance.
This is consistent with \citet{cattaneo2021local}, who show that the equivalent kernel of the LPD estimator with the uniform kernel is the Epanechnikov kernel, i.e., the MSE-optimal kernel.
Among the most commonly used kernels, however, MSE is not very sensitive to kernel choice at interior points, echoing the usual observation for standard KDE.

In inference, by contrast, the Laplace kernel remains the best among the eight kernels, even at interior points.
Thus, for statistical inference, the Laplace kernel is a good choice regardless of the evaluation point, although its numerical advantage is less pronounced than in the boundary case.

Figure \ref{fig: variability interior} shows how the asymptotic variance increases with the polynomial order.
As in standard local polynomial regression (\citealp[p.~79]{Fan_Gijbels:1996}), the variance remains unchanged when moving from odd to even order polynomials.
This observation justifies the common practice of choosing $p=2$.
\begin{figure}[th]
\begin{center}
\begin{tabular}{cc}
      \begin{minipage}[t]{0.45\hsize}
        \centering
        \includegraphics[keepaspectratio, scale=0.45]{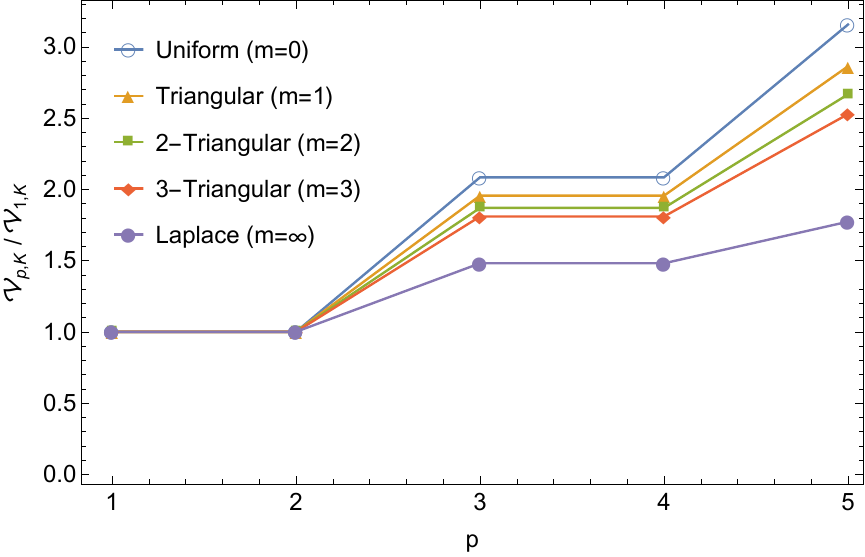}
      \end{minipage} &
      \begin{minipage}[t]{0.45\hsize}
        \centering
        \includegraphics[keepaspectratio, scale=0.45]{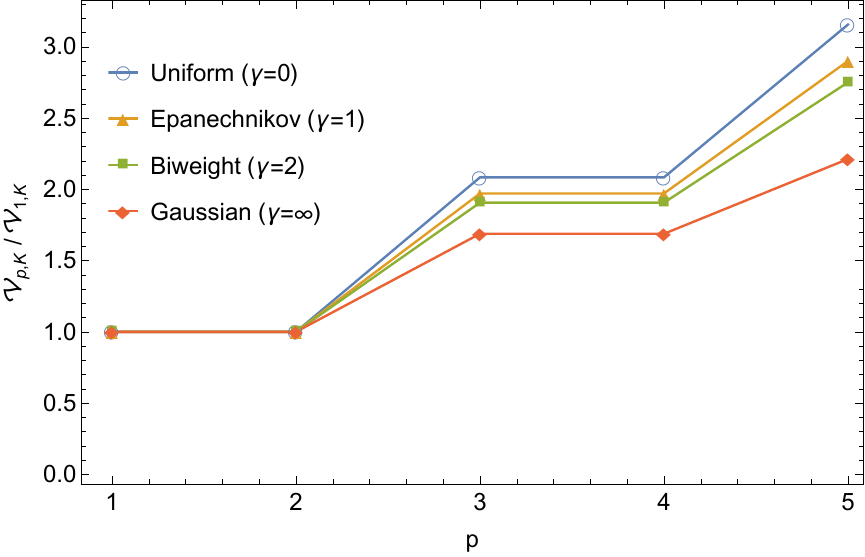}
      \end{minipage} 
    \end{tabular}
     \caption{Increase of Variability at Interior Points}
     \label{fig: variability interior}
\end{center}
\end{figure}

\section{Additinal Numerical Results}\label{sec: oa fig}
\subsection{Small Sample Performance}
Figure \ref{fig: finite mse} reports the MSE under the setup discussed in Section \ref{subsec: simu small} in the main text.
\begin{figure}[H]
    \begin{center}
        \includegraphics[width=0.5\linewidth]{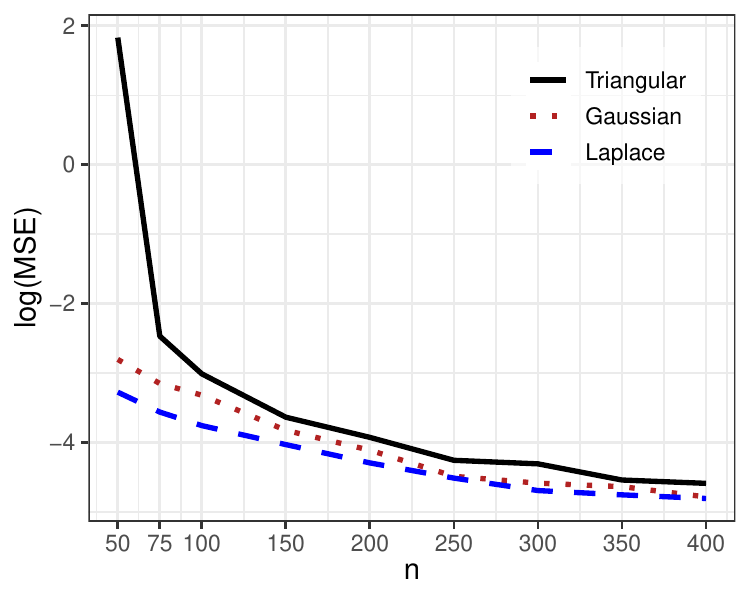}
    \caption{Small-Sample MSE}
     \label{fig: finite mse}
    \end{center}
\end{figure}

\subsection{Increase of Variability}
Figure \ref{fig: variability2} is a counterpart of Figure \ref{fig: variability} for the Beta family kernels.
\begin{figure}[H]
    \begin{center}
        \includegraphics[width=0.6\linewidth]{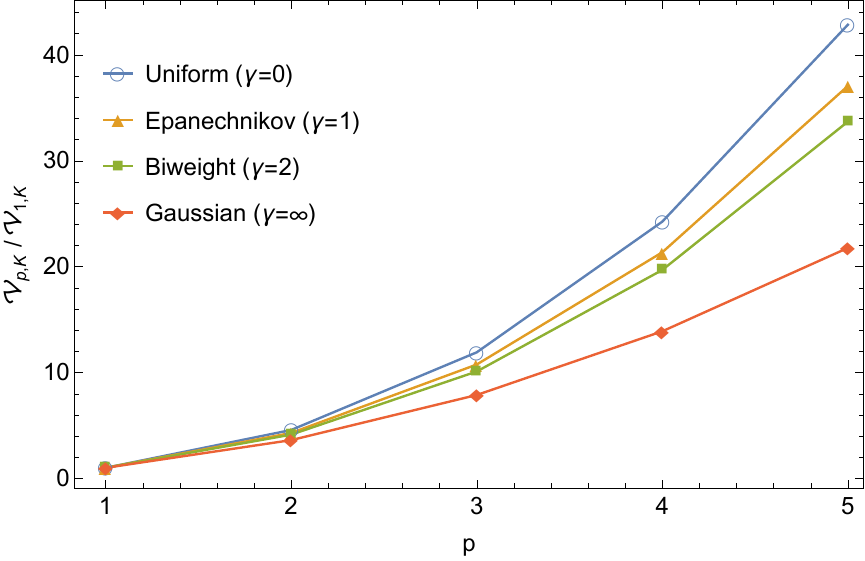}
    \caption{Increase of Variability}
     \label{fig: variability2}
    \end{center}
\end{figure}

\end{document}